\documentclass[a4paper,fleqn]{cas-sc}

\def\tsc#1{\csdef{#1}{\textsc{\lowercase{#1}}\xspace}}
\tsc{WGM}
\tsc{QE}
\tsc{EP}
\tsc{PMS}
\tsc{BEC}
\tsc{DE}
\usepackage{enumitem}
\usepackage{url}
\usepackage{amsmath}
\usepackage{booktabs}
\usepackage{multirow}
\usepackage{tabularx}
\usepackage{amsthm}
\usepackage{amssymb}
\usepackage{colortbl}
\usepackage{xcolor}
\usepackage{graphicx}
\usepackage[ruled,vlined,linesnumbered]{algorithm2e}
\usepackage{mathtools}
\usepackage{epigraph} 
\usepackage[utf8]{inputenc}
\usepackage{CJKutf8}
\usepackage[encapsulated]{CJK}

\usepackage{caption}

\newtheorem{definition}{Definition}
\newtheorem{reductionrule}{Reduction Rule}
\newtheorem{claim}{Claim}
\newtheorem{openquestion}{Open Question}

\newcommand{\spane}[2]{\partial_{#2}(#1)}
\newcommand{\wi}{\textsf{W[}i\textsf{]}}
\newcommand{\bigo}[1]{O(#1)}
\newcommand{\bigos}[1]{O^*(#1)}
\newcommand{\prob}[1]{{\textsc{#1}}}
\newcommand{\mergemanipulation}{\prob{CD-Mani-Merge}}
\newcommand{\mergemanipulationinf}{\prob{CD-Mani-Merge$^{\infty}$}}
\newcommand{\mergeimprovement}{\prob{CD-Improve-Merge}}
\newcommand{\addimprovement}{\prob{CD-Mani-Add}}
\newcommand{\deleteimprovement}{\prob{CD-Mani-Del}}
\newcommand{\addimprovementinf}{\prob{CD-Mani-Add$^{\infty}$}}
\newcommand{\deleteimprovementinf}{\prob{CD-Mani-Del$^{\infty}$}}
\renewcommand{\ss}{\beta}
\newcommand{\yes}{\mbox{Yes}}
\newcommand{\no}{\mbox{No}}
\newcommand{\yesins}{{\yes}-instance}
\newcommand{\noins}{{\no}-instance}
\newcommand{\poly}{\textsf{P}}

\newcommand{\np}{\textsf{\mbox{NP}}}
\newcommand{\nph}{{\np}-{hard}}

\newcommand{\nphns}{\np-hardness}
\newcommand{\npc}{{\np}-{complete}}
\newcommand{\npcns}{\np-completeness}

\newcommand{\wa}{\textsf{\mbox{W[1]}}}

\newcommand{\wah}{{\textsf{W[1]}}-hard}
\newcommand{\wbh}{{\textsf{W[2]}}-hard}

\newcommand{\wbhns}{{\textsf{W[2]}}-hardness}
\newcommand{\fpt}{\textsf{\mbox{FPT}}}
\newcommand{\paranph}{para{\nph}}
\newcommand{\paranphns}{para{\nphns}}

\newcommand{\thms}[1]{{{(Thm.~\ref{#1})}}}
\newcommand{\cors}[1]{{{(Cor.~\ref{#1})}}}
\newcommand{\xp}{\textsf{XP}}
\newcommand{\abs}[1]{\vert #1\vert}
\newcommand{\edge}[2]{\{#1, #2\}}

\newcommand{\wahns}{\textsf{W[1]}-hardness}

\newcommand{\onlyfull}[1]{}
\newcommand{\onlyconf}[1]{}
\newcommand{\remove}[1]{}

\newcommand{\fnotion}[1]{\textit{#1}}
\newcommand{\mergv}[1]{v_{#1}}
\newcommand{\wrt}{with respect to}
\newcommand{\inneighbor}[2]{\Gamma_{#2}^-(#1)}

\newcommand{\outneighbor}[2]{\Gamma_{#2}^+(#1)}

\newcommand{\neighbor}[2]{\Gamma_{#2}(#1)} 
\newcommand{\neighborc}[2]{\Gamma_{#2}[#1]}

\newcommand{\myuplus}{\uplus}
\newcommand{\citeF}[2]{N_{\text{F}}^{#2}(#1)} 
\newcommand{\nciteF}[2]{n_{\text{F}}^{#2}(#1)} 
\newcommand{\citeB}[2]{N_{\text{B}}^{#2}(#1)} 
\newcommand{\nciteB}[2]{n_{\text{B}}^{#2}(#1)} 
\newcommand{\citeR}[2]{N_{\text{R}}^{#2}(#1)} 
\newcommand{\nciteR}[2]{n_{\text{R}}^{#2}(#1)} 

\newcommand{\vset}[1]{V(#1)} 
\newcommand{\eset}[1]{E(#1)} 
\newcommand{\aset}[1]{A(#1)}  
\newcommand{\cd}{CD}
\newcommand{\nciteX}[2]{n_{\text{X}}^{#2}(#1)}
\newcommand{\citeX}[2]{N_{\text{X}}^{#2}(#1)}

\newcommand{\dindex}[2]{\textsf{CDI}_{#2}(#1)}
\newcommand{\setmid}{:}
\newcommand{\bs}{B}
\newcommand{\rs}{R}
\newcommand{\fno}[1]{{\it{#1}}} 

\newcommand{\arc}[2]{(#1,#2)}
\newtheorem{theorem}{Theorem}
\newtheorem{lemma}{Lemma}
\newtheorem{corollary}{Corollary}

\newcommand{\mergo}[3]{#1^{#2}}

\let\shortcite\cite

\renewcommand{\EP}[3]{
\begin{center}
\renewcommand{\tabcolsep}{0.5mm}
{
\begin{tabular}{|lp{0.88\columnwidth}|}\hline
\multicolumn{2}{|l|}{{#1}} \\ \hline
{\bf Given:}    & #2  \\
{\bf Question:} & #3  \\ \hline
\end{tabular}
}
\end{center}
}

\newcommand{\EPP}[4]{
\noindent\fbox{\begin{minipage}{0.95\textwidth}
\noindent\prob{#1}
\newline\noindent(\prob{#2})\hrule
\begin{description}
\item[Given:]  #3
\item[Question:] #4
\end{description}
\end{minipage}
}
}

\begin{document}
\let\WriteBookmarks\relax
\def\floatpagepagefraction{1}
\def\textpagefraction{.001}
\shorttitle{How Hard Is It to Impact the Impact of Your Paper?}
\shortauthors{Yongjie Yang}

\title{How Hard Is It to Impact the Impact of Your Paper?}                      
\tnotemark[1]

\tnotetext[1]{A preliminary version of the paper appeared in the Proceedings of the 33rd International Joint Conference on Artificial Intelligence (IJCAI 2024)~\protect\cite{YangIJCAI2024}. The current version includes all omitted proofs, and all results for {\mergemanipulationinf} are new.}


\author
{Yongjie Yang}[
                        orcid=0000-0002-7731-6818]
\ead{yyongjiecs@gmail.com}

\credit{Conceptualization of this study, Methodology, Software}

\affiliation
{organization={Department of Economics, 
Saarland University},
                city={Saarbr\"{u}cken},
                postcode={66123}, 
                country={Germany}}



\begin{abstract}
Consolidation-disruption index ({\cd} index) is a metric for qualitatively measuring the contribution of a patent or a research paper. Since its inception, a plethora of experimental studies have been undertaken to explore this index. 
We embark on the study of the complexity of {\cd} index manipulation problems, which model scenarios where a scholar aims to enhance the {\cd} indices of their papers through merging, adding, or removing papers.   
We show that these problems are computationally hard, even when restricted to very realistic special cases. Additionally, we explore how different parameters influence the parameterized complexity of these problems. 
\end{abstract}




\maketitle

\epigraph{A craftsman must sharpen his tools to do his job\\ \begin{CJK*}{UTF8}{gbsn}\scriptsize{(工欲善其事，必先利其器)}\end{CJK*}}{\textit{Confucius (551--479 BC)}}


\section{Introduction}
Quantitatively evaluating the achievements of researchers plays a significant role in university recruiting, awarding, grant proposal determination, and more (see, e.g., discussions in~\cite{LbrahimLZRScientific-report2025}). Various measures have been proposed, including the h-index, the i$10$-index, and others. These measures consider a researcher's published papers and the number of citations these papers receive, assigning a numerical score where a higher score is deemed more favorable. 

Another relatively recent yet impactful bibliometric is the consolidation-disruption (CD) index~\cite{DBLP:journals/mansci/FunkO17}. Unlike the aforementioned author-level indices, the CD index is calculated at the paper level.  
According to this measure, the contribution of a research paper can fall into the categories of consolidating, disruptive, or somewhere in between. Consolidating implies that a paper adds value to a field by integrating, improving, or expanding upon existing knowledge. Conversely, disruptive indicates that a paper delves into relatively novel directions that may deviate from preceding research, identifies noteworthy flaws in prior studies, or entirely challenges earlier works. 

To compute the CD index of a paper~$p$ in a citation digraph~$D$ (a directed graph whose vertices represent papers and whose arcs represent citations), three sets~$N^{D}_{\text{F}}(p)$,~$N_{\text{B}}^D(p)$, and~$N^D_{\text{R}}(p)$ are identified\footnote{These notations without the superscript~$D$ are commonly used in the literature. The paper whose CD index is analyzed is typically called the focal paper. The subscripts~$\text{F}$, $\text{B}$, and $\text{R}$ stand for ``focal'', ``both'', and ``remaining'', respectively.}:
\begin{itemize}
    \item $N^D_{\text{F}}(p)$ is the set of nonreferences of~$p$ (papers not cited by~$p$), each of which cites~$p$ but none of~$p$'s references. 
    \item $N^D_{\text{B}}(p)$ is the set of nonreferences of~$p$, each of which cites~$p$ and at least one reference of~$p$.
    \item $N^D_{\text{R}}(p)$ is the set of nonreferences of~$p$, each of which does not cite~$p$ but cites at least one reference of~$p$.
\end{itemize}
The CD index of~$p$ is defined as 
$\frac{\abs{N^D_{\text{F}}(p)}-\abs{N^D_{\text{B}}(p)}}{\abs{N^D_{\text{F}}(p)}+\abs{N^D_{\text{B}}(p)}+\abs{N^D_{\text{R}}(p)}}$. See Figure~\ref{fig-cd-index}. 
\begin{figure}
    \centering
    \includegraphics[width=0.65\textwidth]{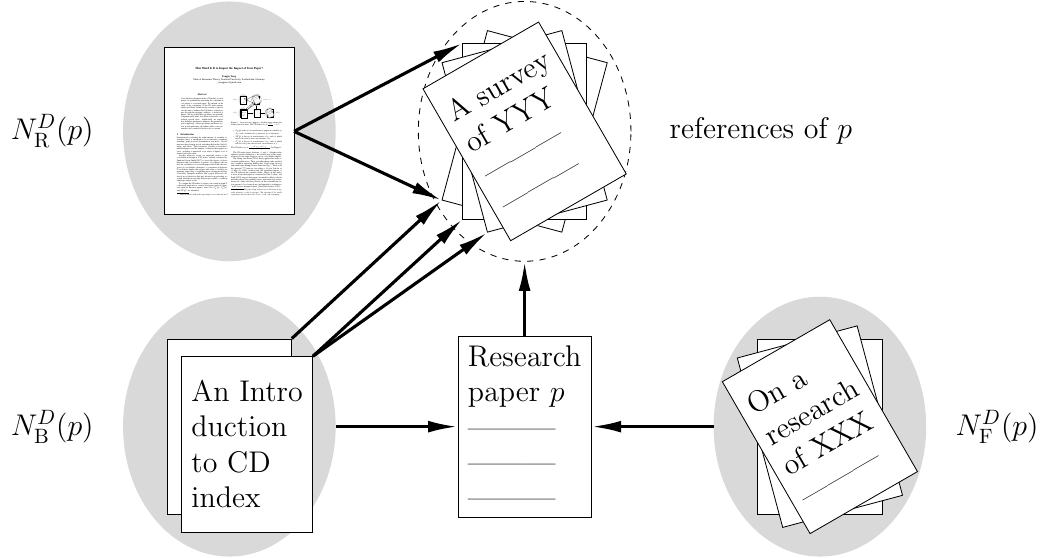}
    \caption{
    An arc from one paper to another indicates that the former cites the latter.
In this example, the CD index of~$p$ equals
$\frac{3-2}{3+2+1}=1/6$}
    \label{fig-cd-index}
\end{figure}
The {\cd} index ranges between~$-1$ and~$1$. A higher value indicates greater disruption to the related area by the paper, whereas a lower value suggests a more consolidating impact~\cite{DBLP:journals/mansci/FunkO17}.
If a paper~$p$ has a {\cd} index of~$1$, then all papers citing~$p$ do not cite any references of~$p$. 
Conversely, if the {\cd} index of~$p$ is~$-1$, then every paper citing~$p$ also cites at least one reference of~$p$.
A paper with a {\cd} index of~$0$ is considered neutral.
In this case, the two sets~$\citeB{p}{D}$ and~$\citeR{p}{D}$ contain the same number of papers. 
See Figure~\ref{fig-cd-index-b} for an illustration.
\begin{figure}
    \centering
    \includegraphics[width=0.95\textwidth]{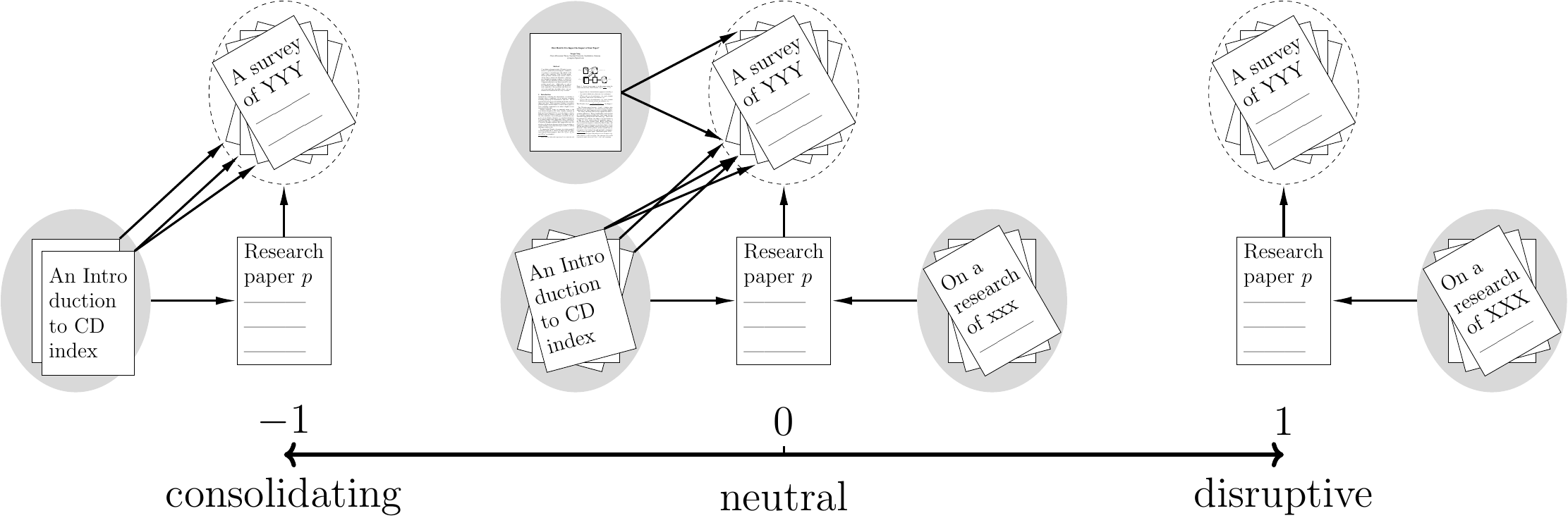}
    \caption{Illustration of citation patterns yielding {\cd} indices of~$-1$,~$0$, and~$1$ for a focal paper~$p$.}
    \label{fig-cd-index-b}
\end{figure}

The {\cd} index was initially examined by Funk and Owen-Smith~\cite{DBLP:journals/mansci/FunkO17} in 2017 to assess the degrees of destabilization and consolidation of patents. 
Subsequently, Wu, Wang, and Evans~\cite{WuWangEvansNature2019} first applied this index to scientific publications in 2019.
Their groundbreaking study revealed the somewhat surprising finding that ``large teams develop and small teams disrupt science and technology''. This influential work has amassed over 1470 citations (according to Google Scholar as of July 15, 2026), underscoring the significant impact of the CD index in the scientific realm. Based on this index, a more recent experimental investigation  conducted by Park, Leahey, and Funk~\cite{ParkLFNature2023} suggests that papers are much less likely to break with the past in ways pushing science and technology in new directions. 
Other scholars pointed out that scientific innovation appears to be slowing down, and appealed to academia to ``Make science disruptive again''~\cite{YanaiLe2023disruptiveagain}.  

However, concerns regarding the reliability of the {\cd} index have prompted public debate
on its application. For example, as noted by Liu, Zhang, and Li~\shortcite{DBLP:journals/joi/LiuZL23},
even minor alterations to a paper's references can significantly affect its {\cd} index.
Such observations have led several researchers to question the validity of the index and to
call for a more rigorous investigation of its theoretical underpinnings before its widespread
adoption~\cite{DBLP:journals/scientometrics/LeibelBornmann2024}.

Motivated by these concerns, we study the manipulability of the {\cd} index from a
computational perspective. Specifically, we investigate the complexity of manipulating
{\cd} indices via operations such as paper merging, addition, and deletion, through the
lens of parameterized complexity. Understanding the computational limits of such
manipulations is essential both for assessing the robustness of the {\cd} index and for
informing the design of more manipulation-resistant bibliometric measures. 
Our problem formulations capture scenarios in which a scholar seeks to improve the
{\cd} indices of a selected set of papers by performing one of the aforementioned
modification operations. The practical relevance of these operations is underscored by
their availability on academic online platforms such as Google Scholar, which allow
scholars to add new manuscripts, delete articles, and merge their own papers. 
Our main contributions are as follows.

\begin{enumerate}
\item[(1)] We initiate the study of {\cd} index manipulation problems, and show that they are computationally hard, even when restricted to very realistic special cases. 

\item[(2)] In the problem of {\cd} index manipulation by merging papers (denoted as {\mergemanipulation}), a scholar seeks to ensure that at least~$\ell$ papers in a citation digraph attain a CD index of at least~$d$ (a threshold deemed satisfactory by the scholar), by performing at most~{$\ss$} merging operations on papers from a given set~$W$, such as papers authored by the scholar. We represent the relationships between papers using a compatibility graph~$H$, where an edge indicates similarity, and only similar papers (forming a clique) are eligible for merging. We show that with respect to the parameters~$\ss$ and~$\ell$,  this problem exhibits {\wahns} or {\paranphns} lower bounds, even when each connected component of~$H$ contains two papers and the citation digraph is acyclic. To complement these results, we also establish {\xp} upper bounds for the problem with respect to the two parameters.

\item[(3)] In the problems of {\cd} index manipulation by adding/deleting papers (denoted as {\addimprovement}/{\deleteimprovement}), a scholar aims to ensure that each paper in a given set~$J$ has a {\cd} index of at least~$d$ by adding/deleting at most~$k$ papers. We explore the parameters~$k$ and~$\abs{J}$. The results for these problems exhibit a diverse range, including both fixed-parameter tractability and fixed-parameter intractability. It is noteworthy that the complexity varies in certain instances when considering cases where $d=1$ and $d<1$.
Moreover, we identify several complexity dichotomies with respect to the parameter~$\abs{J}$. Another interesting finding is that when $d=1$ and the parameter is~$k$, the problem of adding papers is {\wbh}, while the problem of deleting papers is fixed-parameter tractable ({\fpt}).

\item[(4)] We also explore variants of the manipulation problems involving paper merging, addition, and deletion, where the scholar does not concern themselves with the specific number of papers merged, added, or deleted. Our results indicate that, for paper merging and addition, the complexity of the original problem and the variant coincide. In contrast, for paper deletion, the variant proves to be computationally easier than the original formulation.
\end{enumerate}

Our concrete results are summarized in Table~\ref{tab-summary-merge} and Table~\ref{tab-summary-add-deletion}.

The remainder of the paper is organized as follows. Related work is discussed in Section~\ref{sec-related-works}. In Section~\ref{sec-preliminaries}, we introduce preliminary notations, concepts, and background materials that facilitate our study. Section~\ref{sec-problem-formulation} presents the formal definitions of the {\cd} index manipulation problems.
We then present our main results across three dedicated sections, each focusing on one type of modification operation. Specifically, we address manipulation via paper merging in Section~\ref{sec-merging}, paper addition in Section~\ref{sec-addition}, and paper deletion in Section~\ref{sec-deletion}. 
Finally, we conclude the paper in Section~\ref{sec-conclusion}, where we first summarize our main contributions from a retrospective perspective, and then highlight several promising directions for future research.

\begin{table}
 \caption{The parameterized complexity of {\cd} index manipulation by merging papers. The problem is {\fpt} when parameterized by~$\abs{W}$
via brute-force enumeration. 
    ``Acyclic'' (respectively, ``small'') indicates that the results hold even when the citation digraph is acyclic (respectively, when each connected component of the compatibility graph has constant size).}
    \label{tab-summary-merge}
\renewcommand{\tabcolsep}{1.5mm}
    \centering
{
    \begin{tabular}{|l|l|l|l|l|} \hline
\multicolumn{2}{|l|}{} & \multicolumn{2}{l|}{{\mergemanipulation}} & {\mergemanipulationinf}\\ \cline{3-5}
    
        \multicolumn{2}{|l|}{} & $\ss$ & $\ell$ & $\ell$\\ \hline
acyclic
  &   $d\leq 1$
  & $d<1$, $\beta=1$ {\nph} {\thms{thm-merge-d-less-1-ell-1-np-hard}}
  & {$\ell=1$}: {\nph} {\thms{thm-merge-d-1-np-hard}}
  & {$\ell=1$}: {\nph} {\cors{cor-merge-infty-d-1-np-hard}} \\

  &  
  & $d=1$, {$\beta=3$}: {\nph} {\thms{thm-merge-d-1-np-hard}}
  &
  &  \\\hline

acyclic     & $d<1$ &{\wah} (Thm.~\ref{thm-merge-ell-1-hard-small-compatibility}), {\xp} (Thm.~\ref{thm-merge-xp-ss-small-acycli})
        & {$\ell=1$}: {\nph} \thms{thm-merge-ell-1-hard-small-compatibility}
        & {$\ell=2$}: {\nph} ({Thm.~\ref{cor-merge-mani-inf-np-h-j-2-small}})\\ \cline{2-5}

\& small   & $d=1$
&{\wah} (Thm.~\ref{thm-d-1-merge-np-hard-small-comparability}), {\xp} (Thm.~\ref{thm-merge-xp-ss-small-acycli})
&{\wah} (Thm.~\ref{thm-d-1-merge-np-hard-small-comparability}), {\xp} (Thm.~\ref{thm-merge-poly-constant-compatibility-ell-const})
& 
{\xp} {\cors{cor-merge-poly-constant-compatibility-ell-const}} \\ \hline
    \end{tabular}
    }
\end{table}

\begin{table}[ht!]
  \caption{
  The complexity of {\cd} index manipulation by adding or deleting papers. All hardness results except the one for {\deleteimprovementinf} with $d<1$ hold for acyclic citation digraphs. The tractability result for {\deleteimprovementinf} with $d<1$ and $\abs{J}=2$ also holds for acyclic citation digraphs, while all remaining tractability results hold for general citation digraphs.}
    \label{tab-summary-add-deletion}
    \centering
{
    \begin{tabular}{|l|l|l|l|}\hline
         & \multicolumn{2}{c|}{\addimprovement} &  {\addimprovementinf} \\ \cline{2-4}

         & $k$ & $\abs{J}$ & $\abs{J}$ \\ \hline

 $d=1$
 & {\wbh} \thms{thm-add-wbh-k-all-d}
 &{\wah} \thms{thm-add-wah-j-k}
 & {\wah} \thms{thm-add-inf-wah-j-d-1}\\

& {\xp} (trivial)
 &{\xp} \thms{thm-d-1-add-xp-j}
 &{\xp} \thms{thm-d-1-add-xp-j}\\ \hline

 $d<1$
 &{\wbh} \thms{thm-add-wbh-k-all-d}
 & {$\abs{J}=1$}: {\poly} \thms{them-add-p-j-1}
 & {$\abs{J}=1$}: {\poly} \thms{them-add-p-j-1}\\ 

 &{\xp} (trivial)
 & {$\abs{J}=2$}: {\nph} \thms{thm-add-wah-k-j-2-d-smaller}
 & {$\abs{J}=2$}: {\nph} \thms{cor-add-inf-np-hard-j-2-d-smaller}\\ \hline \hline

         & \multicolumn{2}{c|}{\deleteimprovement} &  {\deleteimprovementinf} \\ \cline{2-4}

         &  $k$ & $\abs{J}$ & $\abs{J}$ \\ \hline

 $d=1$
 & {\fpt} \thms{thm-del-d-1-fpt-k}
 & $\abs{J}=1$: {\poly} \thms{thm-del-d-1-p-j-1}
 & {\wah} \thms{thm-d-1-delete-inf-wah-ell}\\

 &
 & $\abs{J}=2$: {\nph} \thms{thm-del-nph-j-2-d-1}
 & {\xp} \thms{thm-del-d-1-xp-j}\\ \hline

 $d<1$
 & {\wah} \thms{thm-delete-improvement-wah}
 & $\abs{J}=1$: {\nph} \thms{thm-delete-improvement-j-paranp-hard}
 & $\abs{J}\leq 2$: {\poly} \thms{thm-inf-poly-j-2}\\ 

 & {\xp} (trivial)
 &
 & $\abs{J}=3$: {\nph} \thms{del-inf-np-hard-j-3}\\ \hline
    \end{tabular}
    }
\end{table}

\section{Related Works}
\label{sec-related-works}
Our study is closely related to and greatly inspired by the works of van Bevern~et~al.~\shortcite{DBLP:journals/qss/BevernKMNSW20,DBLP:journals/ai/BevernKNSW16}, who are the first to explore the complexity of determining if a scholar could improve their h-index by merging, adding, or deleting papers\footnote{The conference versions of the two papers first appeared in 2016 and 2015, respectively.}. In a subsequent paper, Pavlou and Elkind~\shortcite{DBLP:conf/atal/PavlouE16}  studied the complexity of similar problems for the g-index and the i10-index.


From a motivational standpoint, our work also bears relevance to election control problems. These problems model scenarios where an election controller endeavors to influence the election outcome by adding/deleting a limited number of voters/candidates~\cite{Bartholdi92howhard}\footnote{In their seminal paper~\cite{Bartholdi92howhard}, which marked the first systematic exploration of election control problems through a complexity lens, Bartholdi, Tovey, and Trick formally defined the constructive control by adding an unlimited number of candidates. However, they astutely recognized the significance of another variant, wherein the addition of candidates is constrained. Remarkably, their rigorous analysis of the computational complexity of control by adding candidates for the Plurality rule works for both models. 
Their decision to investigate the unlimited candidate addition model stems from its alignment with the investigation of whether a given voting rule adheres to the extensively studied Weak Axiom of Revealed Preference (WARP). 
}%
. We refer to~\cite{BaumeisterR2016,DBLP:conf/ijcai/00010NRSS25,handbookofcomsoc2016Cha7FR} for comprehensive surveys on this line of research.

The {\cd} index can be mathematically regarded as a ``centrality measure'' of vertices in a network. Problems involving improving the influence of a vertex by adding/deleting vertices or edges (or arcs) have been thoroughly explored in the literature. Betzler~et~al.~\shortcite{DBLP:journals/algorithmica/BetzlerBBNU14} investigated the problems of determining whether a given vertex can achieve a maximum/minimum degree by adding/deleting vertices. Variants of the problems restricted to special graph classes have also received attention~\cite{DBLP:journals/fcsc/LiLYY23,DBLP:journals/jda/MishraPD15}. Analogous problems concerning other centrality measures have been studied as well~\cite{DBLP:journals/jea/BergaminiCDMSV18,DBLP:conf/aaai/DAngeloOS19,DBLP:journals/entcs/DAngeloSV16,DBLP:journals/jair/KaczmarekRT25}. 
Furthermore, concepts for quantifying the centrality of groups of vertices and associated centrality-maximizing or -minimizing problems have also been examined~\cite{DBLP:journals/tist/WaniekMWR22,DBLP:journals/tkde/WaniekWZVMR23}. 
For comprehensive summaries of many centrality measures, please refer to~\cite{DBLP:journals/snam/DasSP18,Singh2022,DBLP:journals/access/WanMKMC21}. 

Finally, it is worth noting that  several variants of the CD index have emerged~\cite{LeydesdorffTA2021,DBLP:journals/scientometrics/WangZZ23,DBLP:journals/ipm/YangHZWD23}. Moreover, Jiang and Liu~\shortcite{DBLP:journals/scientometrics/JiangL23} proposed an analogous index for measuring the disruption of journals.

\section{Preliminaries}
\label{sec-preliminaries}
We assume familiarity with basic concepts from graph theory~\cite{DBLP:books/sp/BG2018,Douglas2000}
and (parameterized) complexity theory~\cite{DBLP:books/sp/CyganFKLMPPS15,DBLP:conf/lata/Downey12,DBLP:journals/interfaces/Tovey02}. 
For an integer~$i$, we denote by~$[i]$ the set of all positive integers at most~$i$. 
 An $i$-set is a set of cardinality~$i$; in particular, a $2$-set is called a pair.

\subsection{Citation Digraphs and the {\cd} Index}
A {\fno{citation digraph}}~$D$ is a digraph in which each vertex represents
a scientific paper and an arc~$\arc{v}{u}$ indicates that~$v$ cites~$u$.
We denote by~$\vset{D}$ and~$\aset{D}$ the vertex set and arc set of~$D$, respectively. 
For convenience, we use the terms ``vertex'' and ``paper'' interchangeably. 
For a vertex~$v$ in~$D$, the sets of its {\fno{inneighbors}} and {\fno{outneighbors}} are defined, respectively, as  
$\inneighbor{v}{D}=\{u\in \vset{D}\setminus \{v\} \setmid \arc{u}{v}\in \aset{D}\}$   
and
$\outneighbor{v}{D}=\{u\in \vset{D}\setminus \{v\} \setmid \arc{v}{u}\in \aset{D}\}$. 
The set of {\fno{neighbors}} of~$v$ is the union $\neighbor{v}{D}=\inneighbor{v}{D}\cup \outneighbor{v}{D}$. We refer to the papers in~$\outneighbor{v}{D}$ as the references of~$v$, and those in~$\inneighbor{v}{D}$ as the citations of~$v$. 
Moreover, papers in $\inneighbor{v}{D}\setminus\outneighbor{v}{D}$, i.e., those that cite $v$ but are not cited by $v$, are called non-reciprocal citations of~$v$. 
For $S\subseteq \vset{D}$, let $\outneighbor{S}{D}=(\bigcup_{v\in S}\outneighbor{v}{D})\setminus S$ and let $\inneighbor{S}{D}=(\bigcup_{v\in S}\inneighbor{v}{D})\setminus S$. 
For a paper $v\in \vset{D}$, we define the following three sets:
\begin{itemize}
    \item $\citeF{v}{D}=\{u\in \inneighbor{v}{D}\setminus \outneighbor{v}{D} \setmid \outneighbor{u}{D}\cap \outneighbor{v}{D}=\emptyset\}$, the set of non-reciprocal citations of~$v$ that cite none of its references; 
    \item $\citeB{v}{D}=\{u\in \inneighbor{v}{D}\setminus \outneighbor{v}{D} \setmid \outneighbor{u}{D}\cap \outneighbor{v}{D}\neq\emptyset\}$, the set of non-reciprocal citations of~$v$ that cite at least one of its references;
    \item $\citeR{v}{D}=\{u\in \vset{D}\setminus (\neighbor{v}{D}\cup \{v\}) \setmid \outneighbor{u}{D}\cap \outneighbor{v}{D}\neq \emptyset\}$, the set of papers that cite at least one reference of~$v$ but are neither citations nor references of~$v$.
\end{itemize}
For $X\in \{\text{F}, \text{B}, \text{R}\}$, let $\nciteX{v}{D}=\abs{\citeX{v}{D}}$.


\begin{definition}[{\cd} Index]
    Given a citation digraph~$D$, the {\cd} index of a paper $v \in V(D)$  is defined as
    \[\dindex{v}{D}=\frac{\nciteF{v}{D}-\nciteB{v}{D}}{\nciteF{v}{D}+\nciteB{v}{D}+\nciteR{v}{D}},\]
    if the denominator is nonzero, and is undefined otherwise.
\end{definition}

Notice that, for two papers~$v$ and~$v'$ citing each other,\footnote{This scenario might occur in cases where, for instance, two research groups independently discover the same result almost simultaneously, and the authors mutually agree to acknowledge each other's work for the sake of coordinating publication, academic courtesy, or addressing ethical considerations.} when the {\cd} index of~$v$ is analyzed, the paper~$v'$ is treated as a reference and is therefore excluded from the sets $\citeF{v}{D}$, $\citeB{v}{D}$, and $\citeR{v}{D}$. See Figure~\ref{fig-cd-illustration} for an illustration.

\begin{figure}
    \centering
    \includegraphics[width=0.3\textwidth]{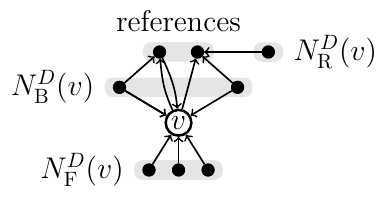}
    \caption{The {\cd} index of $v$ is $\frac{3-2}{3+2+1}=1/6$.}
    \label{fig-cd-illustration}
\end{figure}

A digraph is {\fno{acyclic}} if it contains no directed cycles.
Equivalently, a digraph~$D$ is acyclic if and only if it admits a
{\fno{topological ordering}}; that is, an ordering $(v_1,\dots,v_n)$ of
$\vset{D}$ such that there is no arc $(v_j, v_i)$ with $i < j$ 
(see~\cite[Chapter~3]{DBLP:books/sp/BG2018}).
%

Many of our hardness results hold even when the citation digraph is acyclic. In the corresponding reductions, we explicitly specify topological orderings of the constructed citation digraphs. To this end, for a set $X$, we use $\overrightarrow{X}$ to denote an arbitrary fixed ordering of the elements of $X$. For two orderings $\overrightarrow{X}$ and $\overrightarrow{Y}$,  $(\overrightarrow{X}, \overrightarrow{Y})$ denotes their concatenation. For example, if $\overrightarrow{X} = (x_1, x_2)$, $\overrightarrow{Y} = (y_1, y_2, y_3)$, and $a$ is an element, then $(a, \overrightarrow{X}, \overrightarrow{Y})$ corresponds to the ordering $(a, x_1, x_2, y_1, y_2, y_3)$.


\subsection{The Merging Operation}
Before formally introducing the manipulation problems, we need to clarify a few points.

First, in practice, scholars usually only merge papers with similar titles or topics.
To formulate this assumption, we adopt the notion of {\fnotion{compatibility graph}} used by van Bevern~et~al.~\shortcite{DBLP:journals/ai/BevernKNSW16}.
A compatibility graph is an undirected graph~$H$ where vertices represent papers, and an edge between two papers represents that the two papers share some kind of similarity. A subset of papers can be merged into one paper only if these papers are pairwise adjacent in~$H$.

Second, it is essential to elucidate the {\cd} index of a paper subsequent to certain merging operations. In the study of the h-index manipulation problems, van Bevern~et~al.~\shortcite{DBLP:journals/ai/BevernKNSW16} examined three methodologies to ascertain the citation count a paper accrues post-merging operations. In our framework, not only the quantity but also the categories 
of these citations hold significance.  Hence, the straightforward adoption of the three approaches is not viable. Nevertheless, the fundamental principle of an approach employed by 
van Bevern~et~al.~\cite{DBLP:journals/ai/BevernKNSW16} posits that merging a set of papers entails considering this set as a singular entity. We find this approach most natural and thus incorporate it into our study.

A lingering puzzle remains: how do we handle self-citations? In the majority of prior related research on the {\cd} index, self-citation is typically excluded or encountered  infrequently in experimental works. However, the act of merging papers has the potential to transform a citation digraph that initially lacks self-citation into one where self-citation is present. This transformation is exemplified when a paper is merged with some of its references. Such a scenario is not uncommon, particularly in computer science areas, where merging the journal version and its conference versions occurs, with the journal version often citing the conference versions. A more detailed discussion on self-citations can be found in~\cite{HaunschildBornmann2025}.  
We emphasize that all our hardness reductions are meticulously designed to prevent the occurrence of self-citations, rendering our hardness results irrelevant to this concern. Additionally, our {\xp} algorithms are designed to handle both scenarios.

Having resolved potential confusion, we formalize the aforementioned discussion mathematically.

\begin{definition}[Merging Operation]
\label{def-merging}
A merging operation on a subset of vertices $P \subseteq V(D)$ transforms a citation digraph $D$ into a new digraph $\mergo{D}{P}{M}$ as follows:
\begin{enumerate}
  \item[(1)] Create a new vertex~$\mergv{P}$, set its outneighborhood as~$\outneighbor{P}{D}$, and its inneighborhood as~$\inneighbor{P}{D}$. 
    \item[(2)] Remove all vertices in~$P$ from~$D$.
\end{enumerate}
The vertex $\mergv{P}$ is called the compound vertex of $P$, and each vertex in $P$ is called an atom vertex of $\mergv{P}$.
\end{definition}

\begin{figure}
    \centering
    \includegraphics[width=0.95\textwidth]{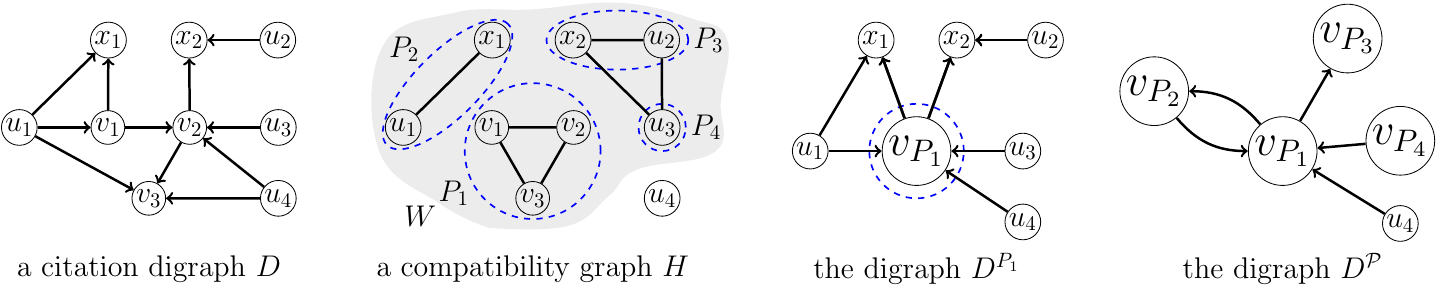}
    \caption{Illustration of Definition~\ref{def-merging}.
Here, $W = V(D)\setminus\{u_4\}$, and
$\mathcal P = \{P_1,P_2,P_3,P_4\}$ is a partition of~$W$ 
that complies with~$H$.
Merging the papers in~$P_1$ yields~$D^{P_1}$,
while merging all papers according to~$\mathcal P$
yields~$D^{\mathcal P}$.
The {\cd} index of the compound vertex~$v_{P_1}$
equals $\frac{2-1}{2+1+1}= \frac14$ in~$D^{P_1}$
and equals~$1$ in~$D^{\mathcal P}$.}
    \label{fig-merge-a}
\end{figure}
For a partition~$\mathcal{P}$ of a subset $W\subseteq V(D)$, let $\mergo{D}{\mathcal{P}}{M}$ denote the citation digraph obtained from~$D$ by performing merging operations on all $P\in \mathcal{P}$ sequentially. Since $\mathcal{P}$ is a partition, the order of these operations has no effect on the final graph. 
For $S\subseteq W$, define $S^{\mathcal{P}}=\{v_{P} \setmid P\in \mathcal{P}, P\cap S\neq\emptyset\}$.
Thus,~$W^{\mathcal{P}}$ is exactly $\{\mergv{P} \setmid P\in \mathcal{P}\}$. 
Let $\mathcal{P}^{\geq 2}=\{P\in \mathcal{P} \setmid \abs{P}\geq 2\}$ be the set of nonsingleton sets in~$\mathcal{P}$. 
We say that~$\mathcal{P}$ {\it{complies with}} a compatibility graph~$H$ if every $P\in \mathcal{P}$ induces a clique in~$H$.
Figure~\ref{fig-merge-a} illustrates the merging operation and related discussion.

\subsection{Parameterized Complexity}
An instance of a {\fnotion{parameterized problem}} is a tuple $(I, \kappa)$, where $I$ is the main part and $\kappa$ is the parameter, typically a nonnegative integer. A parameterized problem is {\fpt} if it can be solved in time $f(\kappa) \cdot \abs{I}^{O(1)}$, and it belongs to  {\xp} if it can be solved in time $\abs{I}^{f(\kappa)}$, where $f$ is a computable function. For a fixed parameter~$\kappa$, both {\fpt} and {\xp} algorithms run in polynomial time. 
A well-established hierarchy of complexity classes is:
\[\fpt\subseteq W[1]\subseteq W[2]\subseteq \cdots\subseteq \xp.\] For $i\geq 1$, a problem is {\textsf{W[$i$]-hard}} if all problems in {\textsf{W[$i$]}} are parameter-reducible to it.
Unless ${\fpt} = {\wi}$, {\textsf{W[$i$]}}-hard problems do not admit {\fpt} algorithms. A problem is {\paranph} if it remains {\nph} even when the parameter is fixed to a constant. Under the assumption that $\text{\poly} \neq \text{\np}$, such problems do not even admit {\xp} algorithms. For further background, we refer the reader to~\cite{DBLP:books/sp/CyganFKLMPPS15,DBLP:conf/lata/Downey12}.

\subsection{Source Problems for Reductions}
This section introduces the problems from which we reduce to establish our hardness results. 

Let $G$ be an undirected graph. The vertex and edge sets of $G$ are denoted by $\vset{G}$ and $\eset{G}$, respectively. An edge between two vertices~$v$ and~$u$ is denoted as $\edge{v}{u}$. The set of {\it{neighbors}} of~$v$ is $\neighbor{v}{G}=\{u\in \vset{G} \setmid \edge{v}{u}\in \eset{G}\}$. 
Every vertex in~$G$ {\it{dominates}} itself and each of its neighbors. A subset~$S\subseteq \vset{G}$ dominates a subset $S'\subseteq \vset{G}$ if every vertex in~$S'$ is dominated by at least one vertex in~$S$. A vertex {\it{covers}} an edge if it is an endpoint of that edge. For $S\subseteq \vset{G}$, let $\spane{S}{G}=\{e\in \eset{G} \setmid e\cap S\neq \emptyset\}$ denote the set of edges covered by $S$. We use~$G[S]$ to denote the {\fno{subgraph}} of~$G$ induced by~$S$, and use $G-S$ to denote the graph obtained from~$G$ by deleting all vertices in~$S$
. 

A subset $S\subseteq \vset{G}$ is a {\it{vertex cover}} of~$G$ if~$S$ covers all edges of~$G$, i.e., $\spane{S}{G}=\eset{G}$. 
An {\fno{independent set}} of~$G$ is a set of pairwise nonadjacent vertices, and a {\fno{clique}} is a set of pairwise adjacent vertices. A $\kappa$-clique is a clique of size~$\kappa$.  
If~$V(G)$ can be partitioned as $(X, Y)$, where both~$X$ and~$Y$ are independent sets, then $G$ is {\fno{bipartite}}, and $(X,Y)$ is called a bipartition of~$G$. Analogously,~$G$ is {\fno{tripartite}} if~$\vset{G}$ can be partitioned into three independent sets.

\EP{\prob{Red-Blue Dominating Set} (\prob{RBDS})}
{A bipartite graph~$G$ with a bipartition $(\rs, \bs)$, and a positive integer~$\kappa\leq \abs{B}$.}
{Is there a subset $\bs'\subseteq \bs$ of cardinality~$\kappa$ that dominates~$R$ in~$G$?}

The vertices in~$\rs$ and~$\bs$ are referred to as {\fno{red vertices}} and {\fno{blue vertices}}, respectively. 
{\prob{RBDS}} is {\nph}~\cite{garey} and is {\wbh} when parameterized by~$\kappa$~\cite{fellows2001}. 

\EP{\prob{Vertex Cover}}
{A graph~$G$ and a positive integer $\kappa\leq \abs{\vset{G}}$.}
{Does $G$ contain a vertex cover of $\kappa$ vertices?}

{\prob{Vertex Cover}} is {\nph} even when restricted to $3$-regular and tripartite graphs~\cite{KartochvilN1990}. In contrast, it is polynomial-time solvable on bipartite graphs~\cite{DBLP:journals/siamcomp/HopcroftK73,Karzanov1973}.

The\onlyfull{{\prob{Independent Set}} and} {\prob{Clique}} problem take the same inputs as {\prob{Vertex Cover}} and asks whether $G$ contains\onlyfull{an independent set or} a $\kappa$-clique\onlyfull{, respectively}. The pproblem is {\nph} and {\wah} when parameterized by $\kappa$, even on regular graphs~\cite{DBLP:journals/cj/Cai08}. 

\EP{\prob{Multicolored Clique}}
{A graph $G$ together with a partition $(V_1,\dots,V_{\kappa})$ of $V(G)$, where $\kappa$ is a positive integer.}
{Does $G$ contain a clique with exactly one vertex from each $V_i$ for $i\in[\kappa]$?}

{\prob{Multicolored Clique}} is {\nph} and {\wah} with respect to~$\kappa$~\cite{DBLP:journals/tcs/FellowsHRV09}. 


\EP{\prob{$3$-Coloring}}
{A graph~$G$.}
{Is $G$ tripartite?
}

{\prob{$3$-Coloring}} is {\nph}~\cite{DBLP:conf/coco/Karp72}.

\section{The Manipulation Problems}
\label{sec-problem-formulation}
We now formally introduce the {\cd} index manipulation problems, starting with the one associated with the paper merging operation. 
This problem models a scenario in which a scholar, allowed to merge a set~$W$ of papers, aims to ensure that at least~$\ell$ of these papers achieve a satisfactory {\cd} index after merging.

\EP{\prob{{\cd} Index Manipulation By Merging Papers} (\mergemanipulation)}
{A citation digraph~$D$, a subset $W\subseteq \vset{D}$, a compatibility graph~$H$ with $\vset{H}=\vset{D}$, two integers~$\ell$,~$\ss$, and a positive rational number~$d\leq 1$.}
{Is there a partition~$\mathcal{P}$ of~$W$ that complies with~$H$ such that $\abs{\mathcal{P}^{\geq 2}}\leq \ss$, and
     at least~$\ell$ papers in~$W^{\mathcal{P}}$ have {\cd} index at least~$d$ in~$\mergo{D}{\mathcal{P}}{\mu}$?
}

Here, the scholar is allowed to perform at most~$\ss$ merging operations ($\abs{\mathcal{P}^{\geq 2}}\leq \ss$). We also study a natural variant,  {\mergemanipulationinf}, where this restriction is lifted: the problem  asks whether there exists a partition~$\mathcal{P}$ of~$W$ that complies with~$H$ such that at least~$\ell$ papers in~$W^{\mathcal{P}}$ have  {\cd} index at least~$d$ in~$\mergo{D}{\mathcal{P}}{\mu}$, without bounding the number of merge operations.

The subsequent manipulation problems model a scenario where a scholar has selected a set~$J$ of papers published in reputable venues, or papers considered robust according to other metrics. This selection aims to strengthen her application for a research grant (or a scientific position, etc.). Some of these papers exhibit relatively low {\cd} indices. The grant committee, prioritizing projects with innovative ideas or those opening new directions, highly values such aspects. Assume that an online platform, endorsed by the granting entity, offers {\cd} index calculation functions, and that scholars are permitted to add new papers (in a set denoted~$U$) or remove existing ones (in a set denoted~$V$) from the system. To increase the chance of success, the scholar seeks to determine if she can elevate the {\cd} indices of all selected papers to a satisfactory level by adding or deleting at most~$k$ papers. 
The two problems are formally defined as follows.

For a citation digraph~$D$ and a subset $S\subseteq \vset{D}$,~$D[S]$ is the subgraph of~$D$ induced by~$S$, and $D-S$ denotes the citation digraph obtained from~$D$ by removing all papers in~$S$ (and arcs incident to them).

\onlyfull{
\EP{\mergeimprovement}
{A citation digraph~$D$, two subsets $W\subseteq \vset{D}$ and $J\subseteq W$, a compatibility graph~$H$ such that $\vset{H}=\vset{D}$.}
{$\exists$ a partition~$\mathcal{P}$ of~$W$ that complies with~$H$ such that for each $P\in \mathcal{P}$ it holds that $\abs{P\cap J}\leq 1$, and for every $P\in \mathcal{P}$ such that $\abs{P\cap J}=1$, it holds that the {\cd} index of the compound paper~$\mergv{P}$ in~$\mergo{D}{\mathcal{P}}{\mu}$ is strictly larger than that of the paper in $P\cap J$ in~$D$?}

In particular, in this setting paper addition, there is a set~$V$ of papers that are already registered in the system (e.g., papers updated in the google scholar profile), and a set~$U$ of papers that have not registered (e.g., newly accepted papers that have not been automatically updated, and but the scholar determine to manually update some of them to the system). As in the above problem, the scholar has a set~$J$ of preferred papers from~$V$, and they aim to improve the {\cd} indices of papers from~$J$ to a threshold value by selecting which papers from~$U$ are updated to the system. In the context of paper deletion, the scholar has no new accepted papers, and thus they strive to achieve their goal by deleting papers from $V\setminus J$. The two problems are formally defined below.
}

\EP{\prob{{\cd} Index Manipulation By Adding Papers} (\addimprovement)}
{A citation digraph~$D$, a tuple $(V, U)$ of two disjoint subsets of papers in~$D$
, a subset $J\subseteq V$, a nonnegative integer~$k$, and a positive rational number~$d\leq 1$.}
{Is there  $U'\subseteq U$ with $\abs{U'}\leq k$ such that every paper in $J$ has {\cd} index at least~$d$ in $D[V\cup U']$?}

\EP{\prob{{\cd} Index Manipulation By Deleting Papers} (\deleteimprovement)}
{A citation digraph~$D$, a subset $J\subseteq \vset{D}$, a nonnegative integer~$k$, and a positive rational number~$d\leq 1$.}
{Is there $V'\subseteq \vset{D}\setminus J$ with $\abs{V'}\leq k$ such that every paper in~$J$ has {\cd} index  at least~$d$ in $D-V'$?}

For clarity, in the context of {\addimprovement}, papers in~$V$ are referred to as {\fno{registered papers}} and those in~$U$ as {\fno{unregistered papers}}. In both {\addimprovement} and {\deleteimprovement}, the papers in~$J$ are called {\fno{distinguished papers}}. 

We also explore natural variants without the addition/deletion budget~$k$. {\addimprovementinf} and {\deleteimprovementinf} take the same input as {\addimprovement} and {\deleteimprovement}, respectively, but without~$k$, and ask whether there exists $U'\subseteq U$ (respectively, $V'\subseteq \vset{D}\setminus J$) such that every paper in~$J$ has {\cd} index at least~$d$ in $D[V\cup U']$ (respectively, $D-V'$).

The following sections establish the complexity of the {\cd} index manipulation problems. 
We remark that all reductions presented in this paper run in polynomial time, and all {\cd} manipulation problems are in {\np}. Consequently, if a problem is shown to be {\wah} or {\wbh} with respect to a given parameter, it is also {\npc}. To avoid repetition, we do not explicitly state {\npcns} results in such cases. 

\section{Manipulation by Paper Merging}
\label{sec-merging}
This section studies the complexity of {\mergemanipulation}. We first show that the problem is computationally hard even in special cases.

\begin{theorem}
\label{thm-merge-ell-1-hard-small-compatibility}
For any positive rational number $d<1$, {\mergemanipulation} is {\wah} with respect to~$\ss$, even if $\ell=1$, each connected component of the compatibility graph contains at most two papers, and the citation digraph is acyclic.
\end{theorem}

\begin{proof}
We prove Theorem~\ref{thm-merge-ell-1-hard-small-compatibility} by a reduction from {\prob{Clique}} on regular graphs to {\mergemanipulation}. We first present the reduction for $d=1/3$ and then explain how to adapt it for any positive rational $d<1$. 

Let $I=(G, \kappa)$ be an instance of {\prob{Clique}}, where $G=(V, E)$ is a $t$-regular graph with~$n$ vertices and $m={n\cdot t}/{2}$ edges. Without loss of generality, we assume $t>\kappa>0$. 
We construct an instance $g(I)=(D, W, H, \ell, \ss, d)$ of {\mergemanipulation} with $\ell=1$, $\ss=\kappa$, and $d=1/3$ as follows.
\begin{description}
    \item[Citation digraph~$D$.] For each vertex $v\in V$, we create two papers~$v(1)$ and~$v(2)$. For each $i\in [2]$,
    let $V(i)=\{v(i) \setmid v\in V\}$.
    For each edge~$e\in E$, we create a paper~$p(e)$. Let~$P(E)=\{p(e) \setmid e\in E\}$.  
    In addition, we create two papers~$v^{\star}$ and~$y$, a set~$X$ of $2(n-\kappa)+t\cdot \kappa-\frac{\kappa\cdot (\kappa-1)}{2}$, a set~$Z$ of $t\cdot \kappa-\frac{\kappa\cdot (\kappa-1)}{2}$ papers, and a set~$S$ of~$4(m+2n+1)$ papers. Since $t>\kappa>0$,~$X$ and~$Z$ are nonempty.  
    The arcs of $D$ are defined by the following citations:
 \begin{itemize}
     \item All papers in~$X$ cite~$v^{\star}$.
     \item The paper~$v^{\star}$ cites all papers in~$V(2)$.
     \item All papers in~$V(1)$ cite all papers in~$V(2)\cup \{v^{\star}\}$.
     \item For each edge $e=\edge{u}{v}$ in~$G$, the paper~$p(e)$ cites both~$u(1)$ and~$v(1)$.
     \item All papers in~$V(1)\cup V(2)\cup S$ cite the paper~$y$.
     \item All papers in~$Z$ cite all papers in $V(1)\cup V(2)\cup \{y\}$.
 \end{itemize}
 The citation digraph~$D$ is acyclic with the topological ordering $(\overrightarrow{X}, \overrightarrow{S}, \overrightarrow{Z}, \overrightarrow{P(E)}, \overrightarrow{V(1)}, v^{\star}, \overrightarrow{V(2)}, y)$. 
 \item[Papers allowed to be merged.] $W=V(1)\cup V(2) \cup \{v^{\star}\}$.
 \item[Compatibility graph~$H$.] The edge set of $H$ is $\{\edge{v(1)}{v(2)} \setmid v\in V\}$.
\end{description}
\begin{figure}
    \centering
    \includegraphics[width=0.5\textwidth]{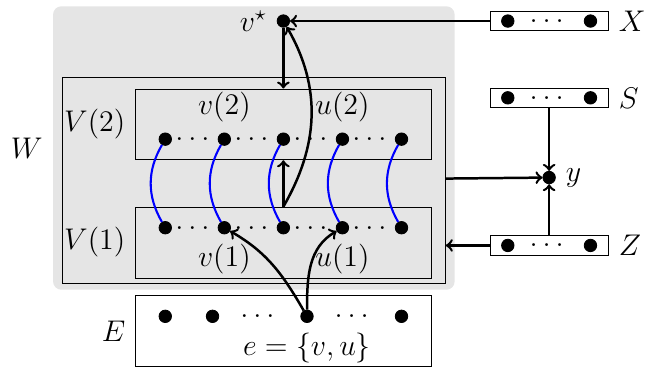}
    \caption{Illustration of the reduction in the proof of Theorem~\ref{thm-merge-ell-1-hard-small-compatibility}. Blue edges represent the compatibility graph $H$, and black arcs represent citations in the digraph $D$.}
    \label{fig-merge-ell-1-hard-small-compatibility}
\end{figure}
The instance $g(I)$ can be constructed in polynomial time. Figure~\ref{fig-merge-ell-1-hard-small-compatibility} illustrates the reduction. We prove its  correctness below.

$(\Rightarrow)$ Assume that~$G$ has a $\kappa$~clique~$K$. Recall that $\spane{K}{G}=\{e\in E \setmid e\cap K\neq \emptyset\}$ denotes the set of edges in~$G$ covered by~$K$. Since~$G$ is $t$-regular, we have $\abs{\spane{K}{G}}=t\cdot \kappa-\frac{\kappa\cdot (\kappa-1)}{2}$. 
 Let~$\mathcal{P}$ be the partition of~$W$ defined by $\mathcal{P}^{\geq 2}=\{\{v(1), v(2)\} \setmid v\in K\}$, which complies with~$H$ and has $\abs{\mathcal{P}^{\geq 2}}=\kappa=\ss$. Consider the citation digraph~$D^{\mathcal{P}}$, obtained from~$D$ by merging~$v(1)$ and~$v(2)$ for each~$v\in K$. The {\cd} index of~$v^{\star}$ in~$D^{\mathcal{P}}$ is
 \[\frac{\abs{X}-(\abs{V(1)}-\abs{K})}{\abs{X}+(\abs{V(1)}-\abs{K})+(\abs{Z}+\abs{\spane{K}{G}})}.\]
 Substituting the values of~$\abs{X}$,~$\abs{V(1)}$,~$\abs{K}$,~$\abs{Z}$, and~$\abs{\spane{K}{G}}$ shows that the CD index reaches the threshold $1/3$.

$(\Leftarrow)$ Assume that there is a partition~$\mathcal{P}$ of~$W$ that complies with the compatibility graph~$H$ such that $\abs{\mathcal{P}^{\geq 2}}\leq \ss=\kappa$ and at least one paper in~$W^{\mathcal{P}}$ has {\cd} index at least~$1/3$ in the citation digraph~$D^{\mathcal{P}}$. The presence of~$y$,~$S$, and~$Z$ prevents any paper in $(V(1)\cup V(2))^{\mathcal{P}}$ from achieving a {\cd} index of at least $1/3$ in $D^{\mathcal{P}}$. Therefore, the only possibility is that~$v^{\star}$ achieves a {\cd} index of at least~$1/3$ in~$D^{\mathcal{P}}$. Let $K=\{v\in V \setmid \{v(1),v(2)\}\in \mathcal{P}^{\geq 2}\}$ be the set of vertices of~$G$ corresponding to merged pairs of papers in~$\mathcal{P}$. Merging each pair~$\{v(1), v(2)\}$ for any $v \in V$ alters the role of~$v(1)$---which originally cites both $v^{\star}$ and some of its references---by removing it as a direct reference of~$v^{\star}$. At the same time, this operation transforms all edge-papers that cite~$v(1)$ into papers that does not cite~$v^{\star}$, but cite at least one reference of~$v^{\star}$. Consequently, the {\cd} index of~$v^{\star}$ in~$D^{\mathcal{P}}$ is
 \begin{equation}
 \label{eq-xx}
     \frac{\abs{X}-(\abs{V(1)}-\abs{\mathcal{P}^{\geq 2}})}{\abs{X}+(\abs{V(1)}-\abs{\mathcal{P}^{\geq 2}})+(\abs{Z}+\abs{\spane{K}{G}})}.
 \end{equation}
Let~$m'$ be the number of edges within~$G[K]$. We have $\abs{\spane{K}{G}}=t\cdot \abs{K}-m'$. The fraction in~\eqref{eq-xx} is at least $1/3$ if and only if $\abs{K}=\abs{\mathcal{P}^{\geq 2}}=\kappa$ and $m'={\kappa \cdot (\kappa-1)}/{2}$. That is,~$K$ contains exactly~$\kappa$ vertices, and there are ${\kappa\cdot (\kappa-1)}/{2}$ edges in~$G[K]$. This implies that~$K$ forms a clique in~$G$, and therefore~$I$ is a {\yesins}.


We now consider the reduction for any positive rational value of~$d$. Let ${p}/{q}$ be the canonical form of~$d$. When $q \geq p + 2$, we adapt the previous construction by extending the three sets~$X$,~$Z$, and~$S$, ensuring that the resulting {\cd} index equals the target $d$. Specifically, we define:
\begin{itemize}
    \item $X$ as a set of $(p+1)\cdot (n-\kappa) + p\cdot t\cdot \kappa - \frac{p\cdot \kappa \cdot (\kappa-1)}{2}$ papers,
    \item $Z$ as a set of  $(q-p-2)\cdot (n-\kappa) + (q-p-1)\cdot t\cdot \kappa - \frac{(q-p-1)\cdot \kappa \cdot (\kappa-1)}{2}$ papers, and
    \item $S$ as a set of $(p + q)\cdot (m + 2n + 1)$ papers.
\end{itemize}
For the case $q = p+1$, the above definition would give $Z$ negative size. Instead, we define: 

\begin{itemize}
    \item $X$ as a set of $(3p+1)\cdot (n-\kappa) + p\cdot t\cdot \kappa - \frac{p\cdot \kappa \cdot (\kappa-1)}{2}$ papers,
    \item $Z$ as a set of $n - \kappa$ papers, and 
    \item $S$ as a set of $(p + 1)\cdot (m + 2n + 1)$ papers.
\end{itemize}

In both cases, correctness follows from the same algebraic arguments as in the case $d=1/3$.
\end{proof}

Theorem~\ref{thm-merge-ell-1-hard-small-compatibility} implies that {\mergemanipulation} is {\paranph} with respect to the parameter~$\ell$. By introducing dummy papers, this hardness result can be extended to any constant integer $\ell > 1$.\footnote{Specifically, one can add $\ell-1$ dummy papers to $J$ and let they be cited by all papers in $X$.}  
This reduction can be further adapted to establish the {\nphns} of {\mergemanipulationinf}, the variant in which the strategic scholar subjectively does not care about the number of merging operations. 

\begin{theorem}
\label{cor-merge-mani-inf-np-h-j-2-small}
For any positive rational number $d<1$, {\mergemanipulationinf} is {\nph}, even if $\ell=2$, each connected component of the compatibility graph contains at most two papers, and the citation digraph is acyclic. 
\end{theorem}

\begin{proof}
We modify the reduction from Theorem~\ref{thm-merge-ell-1-hard-small-compatibility}, first considering $d=1/3$. We introduce an auxiliary paper $u^\star \in W$, which is cited by every paper in $V(1) \cup V(2)$. Additionally, we add a reference $y'$ for $u^\star$ and a set $S'$ consisting of $4n-2\kappa$ papers that cite $y'$. If~$G$ has a $\kappa$-clique, merging the corresponding vertex pairs results in both~$v^\star$ and~$u^\star$ reaching a CD index of $1/3$. Conversely, $v^\star$ and $u^\star$ are the only papers that can potentially reach the threshold by valid merging of papers in $W$. Although the merging budget is omitted, the CD index of $u^\star$ serves as a ``soft budget'': Each merger of papers from $V(1)\cup V(2)$ reduces the number of papers that cite $u^\star$ but do not cite its reference $y'$. If more than $\kappa$ pairs are merged, the denominator in the CD index formula for $u^\star$ becomes too small relative to the remaining citations, causing the index to fall below $1/3$. Generalizing to any rational $d < 1$ follows the scaling logic established in the proof of Theorem~\ref{thm-merge-ell-1-hard-small-compatibility}.
\end{proof}

The hardness result in Theorem~\ref{thm-merge-ell-1-hard-small-compatibility} does not extend to the case $d = 1$. Recall that a paper has {\cd} index~$1$ if and only if, among the three sets defining the {\cd} index, only~$\citeF{v}{D}$ is nonempty. While this condition significantly restricts the merging search space, determining which papers can attain an index of $1$ remains a combinatorial challenge, as established in the following theorem.

\begin{theorem}
\label{thm-d-1-merge-np-hard-small-comparability}
For $d=1$, {\mergemanipulation} is {\wah} with respect to~$\ss+\ell$, even if each connected component of the compatibility graph has at most two papers and the citation digraph is acyclic
.
\end{theorem}

\begin{proof}
    We prove the theorem by a reduction from {\prob{Clique}}. Let $I=(G,\kappa)$ be an instance of {\prob{Clique}} where $G=(V, E)$. We create an instance $g(I)=(D, W, H, \ell, \ss, d)$ of {\mergemanipulation} with $d=1$ as follows. An illustration is provided in Figure~\ref{fig-d-1-merge-np-hard}.  
    \begin{description}
        \item[Citation digraph~$D$.] For each $v\in V$, we create two papers~$v(1)$ and~$v(2)$ in~$D$. For each $i\in [2]$, let $V(i)=\{v(i) \setmid v\in \vset{G}\}$. For each edge $e\in E$, we create one paper~$p(e)$. Let~$P(E)=\{p(e) \setmid e\in E\}$. Additionally, we create three papers~$x$,~$y$, and~$z$. 
        The arcs of $D$ are defined by the following citations:
    \begin{itemize}
        \item The paper~$x$ cites all papers from~$P(E)$.
        \item Each edge-paper~$p(e)$, where $e=\edge{v}{u}$ is an edge in~$G$, cites~$v(1)$ and~$u(1)$.
        \item For each $v\in V$,~$v(2)$ cites~$v(1)$.
        \item All papers from $V(1)\cup V(2)\cup \{z\}$ cite~$y$.
    \end{itemize}
The citation digraph~$D$ is acyclic, admitting the topological ordering
 $(x, z, \overrightarrow{P(E)}, \overrightarrow{V(2)}, \overrightarrow{V(1)}, y)$. 
    \item[Papers allowed to be merged.] $W=V(D)\setminus \{x, y, z\}$.
    \item[Compatibility graph~$H$.] The edge set of $H$ is~$\{\edge{v(1)}{v(2)}\setmid v\in V\}$.
    \end{description}
Finally, we set $\ell={\kappa\cdot (\kappa-1)}/{2}$ and $\ss=\kappa$. 
\begin{figure}
    \centering
    \includegraphics[width=0.4\textwidth]{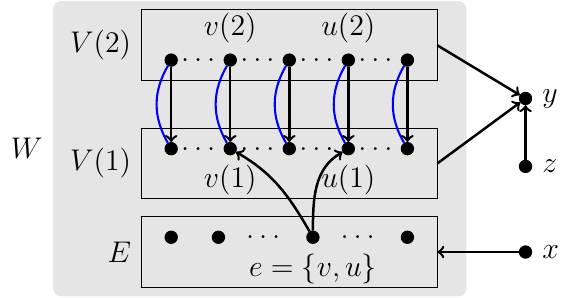}
    \caption{Illustration of the reduction in the proof of Theorem~\ref{thm-d-1-merge-np-hard-small-comparability}. Blue edges represent the compatibility graph $H$, and black arcs represent citations in the digraph $D$.}
    \label{fig-d-1-merge-np-hard}
\end{figure}
We prove the correctness of the reduction as follows.

    $(\Rightarrow)$ If~$G$ contains a $\kappa$-clique~$K$, we merge~$v(1)$ and~$v(2)$ for each $v\in K$. In the resulting citation digraph, all the ${\kappa\cdot (\kappa-1)}/{2}$ edge-papers corresponding to edges of~$G[K]$ attain a {\cd} index of~$1$. Since $\ell={\kappa\cdot (\kappa-1)}/{2}$ and $\ss=\kappa$,~$g(I)$ is a {\yesins}. 

    $(\Leftarrow)$ Assume that~$g(I)$ is a {\yesins}. That is, there exists a partition~$\mathcal{P}$ of~$W$ complying with~$H$ such that $\abs{\mathcal{P}^{\geq 2}}\leq \ss=\kappa$, and at least~$\ell$ papers in~$W^{\mathcal{P}}$ have a {\cd} index of~$1$ in~$D^{\mathcal{P}}$. The presence of~$y$ and~$z$ ensures that no vertex-paper in $(V(1)\cup V(2))^{\mathcal{P}}$ can have a {\cd} index of~$1$ in~$D^{\mathcal{P}}$. Therefore, there exists $E'\subseteq E$ with $\abs{E'}\leq \ell={\kappa \cdot (\kappa-1)}/{2}$ such that each edge-paper~$p(e)$ for $e\in E'$ has a {\cd} index of~$1$ in~$D^{\mathcal{P}}$. 
    For an edge-paper $p(e)$ where $e = \edge{v}{u}$ to attain a {\cd} index of~$1$, we need to merge $v(2)$ with $v(1)$, and merge $u(2)$ with $u(1)$. 
    Let $K=\{v\in V \setmid \{v(1), v(2)\}\in \mathcal{P}^{\geq 2}\}$. By the above discussion,~$E'$ is a subset of edges in the induced subgraph~$G[K]$. Moreover, $\abs{K}= \abs{\mathcal{P}^{\geq 2}}\leq \kappa$ and $\abs{E'}\geq  {\kappa\cdot (\kappa-1)}/{2}$. This scenario is achievable only if $\abs{K}= \kappa$ and $\abs{E'}= {\kappa \cdot (\kappa-1)}/{2}$. Thus,~$K$ is a $\kappa$-clique in~$G$, and~$I$ is a {\yesins}.
\end{proof}



What if the objective is to ensure that a single paper attains a {\cd}-index of~$1$? Does focusing on a single target reduce the problem’s complexity? As Theorem~\ref{thm-merge-d-1-np-hard} shows, this is unlikely---not only for $d=1$, but for all fixed positive values of $d$. 

\begin{theorem}
\label{thm-merge-d-1-np-hard}
For any positive rational number $d \leq 1$, {\mergemanipulation}  is {\nph}, even when $\ell = 1$, $\ss = 3$, and the citation digraph is acyclic with maximum outdegree four.
\end{theorem}

\begin{proof}
    Let~$d\leq 1$ be a positive rational number in canonical form $p/q$. 
We prove the theorem by a reduction from \prob{$3$-Coloring}. Let $G$ be the input graph with vertex set $V$. Without loss of generality, we assume that $G$ is not bipartite. We construct an instance $I=(D, W, H, \ell, \ss, d)$ of {\mergemanipulation} with $\ell=1$ and $\ss=3$ as follows. An illustration of the reduction is shown in Figure~\ref{fig-cd-merge-np-hard-d-1}. 
    Let $\vset{D}=V \cup X\cup X'\cup \{v^{\star}, c_1, c_2, c_3, y\}\cup Z$, where~$\abs{X}=q+p$, $\abs{X'}=q-p$, and $\abs{Z}=q\cdot (q-p+1)$. 
    The arcs are defined as follows:
        \begin{itemize}
        \item All papers from $X$ cite~$v^{\star}$.
        \item All papers from $X'$ cite $v^{\star}$, $c_1$, $c_2$, and $c_3$.
        \item The paper~$v^{\star}$ cites~$c_1$,~$c_2$, and~$c_3$.
        \item All papers from $V$ cite~$c_1$,~$c_2$,~$c_3$, and $y$.
        \item All papers from $\{c_1, c_2, c_3\}\cup Z$ cite~$y$.
    \end{itemize}
  By construction, each paper cites at most four others, and some cite exactly four. Hence, the maximum outdegree of the citation digraph~$D$ is four. Moreover, the graph~$D$ is acyclic, as it admits the following topological ordering: $(\overrightarrow{Z}, \overrightarrow{X}, \overrightarrow{X'}, v^{\star}, \overrightarrow{V}, c_1, c_2, c_3, y)$. 
    Let $W=V\cup\{v^{\star}, c_1, c_2, c_3\}$.
    \begin{figure}[ht!]
    \centering{
       \includegraphics[width=0.35\textwidth]{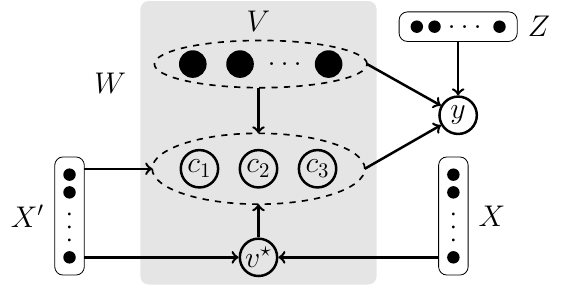}
       }
        \caption{Illustration of the reduction in the proof of Theorem~\ref{thm-merge-d-1-np-hard}.}
         \label{fig-cd-merge-np-hard-d-1}
    \end{figure}
    The edges in the compatibility graph~$H$ include exactly all missing edges in~$G$, and all edges between~$V$ and $\{c_1, c_2, c_3\}$. 
    
    We prove the correctness of the reduction as follows.

    $(\Rightarrow)$ Assume there exists a partition $(U_1, U_2, U_3)$ of~$V$ such that each~$U_i$, $i\in [3]$, is an independent set in~$G$. By the definition of~$H$,~$U_i$ together with any of~$c_1$,~$c_2$, and~$c_3$ forms a clique in~$H$. After merging all papers in $U_i\cup \{c_i\}$ for each $i\in [3]$, the {\cd} index of~$v^{\star}$ is $\frac{\abs{X}-\abs{X'}}{\abs{X}+\abs{X'}}=\frac{p}{q}$. Thus, the partition $\{U_1\cup \{c_1\}, U_2\cup \{c_2\}, U_3\cup \{c_3\}, \{v^{\star}\}\}$ of~$W$ is a {\yes}-witness for~$I$.

    $(\Leftarrow)$ Assume that~$I$ is a {\yesins}. That is, there is a partition~$\mathcal{P}$ of~$W$ that complies with~$H$ and at least one paper from~$W^{\mathcal{P}}$ has a {\cd} index of at least~${p}/{q}$ in~$D^{\mathcal{P}}$. For each $c_i$, $i\in [3]$, let $P_i\in \mathcal{P}$ be the part containing~$c_i$. By the definition of~$H$, only papers in $V\cup \{c_1, c_2, c_3\}$ can be merged. By construction, the {\cd} index of any paper from $\{c_1, c_2, c_3\}^{\mathcal{P}}$ in~$D^{\mathcal{P}}$ can be at most
    $\frac{\abs{X'\cup \{v^{\star}\}}}{\abs{X'\cup \{v^{\star}\}}+\abs{Z}}=\frac{1}{q+1}<d$. In addition, papers from~$\vset{G}$ have no citations in~$D$. Therefore,~$v^{\star}$ is the only paper with a {\cd} index of at least~$d$ in~$D^{\mathcal{P}}$. 
    Note that $\cramped{\citeF{v^{\star}}{D^{\mathcal{P}}}}=X$ and $\citeB{v^{\star}}{D^{\mathcal{P}}}=X'$. It follows that the {\cd} index of~$v^{\star}$ in~$D^{\cramped{\mathcal{P}}}$ is at least ${{p}/{q}}$ if and only if $\citeR{v^{\star}}{D^{\mathcal{P}}}=\emptyset$. As $\citeR{v^{\star}}{D}=V$, this implies that every $v\in V$ is contained in some~$P_i$, $i\in [3]$. Because $\mathcal{P}$ complies with~$H$, and $H$ contains all missing edges of~$G$, each $P_i\setminus \{c_1, c_2, c_3\}$ is an independent set in~$G$. Hence, the original {\prob{$3$-Coloring}} instance is a {\yesins}.
\end{proof}

For $d$ strictly smaller than one, hardness holds even when $\ell=\beta=1$.

\begin{theorem}
\label{thm-merge-d-less-1-ell-1-np-hard}
For any positive rational number $d < 1$, {\mergemanipulation}  is {\nph}, even when $\ell = \ss = 1$, and the citation digraph is acyclic with maximum outdegree two.
\end{theorem}

\begin{proof}
    Let~$d< 1$ be a positive rational number in canonical form $p/q$. 
We prove the theorem by a reduction from \prob{Clique}. Let $(G, \kappa)$ be an instance of {\prob{Clique}}, where $G=(V, E)$ and $n=\abs{V}$. We construct an instance $I=(D, W, H, \ell, \ss, d)$ of {\mergemanipulation} with $\ell=\ss=1$ as follows. 
    \begin{description}
        \item[Citation digraph~$D$.] For each vertex $v$ in $G$, we create a set $P(v)=(v(1), v(2), \dots, v({q-p}))$ of $q-p$ papers. Let $U$ denote the set of all these $n\cdot (q-p)$ papers. 
       Additionally, we create two papers $c$ and $y$, a set~$X$ of $p\cdot (n-\kappa)$ papers, and a set $Z$ of $p+q$ paper. 
    
    The arcs are defined as follows:
        \begin{itemize}
        \item All papers in $X$ cite~$v^{\star}$.
        \item The paper~$v^{\star}$ cites~$c$.
        \item All papers in $U$ cite both~$c$ and $y$.
        \item All papers in $\{c\}\cup Z$ cite~$y$.
    \end{itemize}
  By construction, each paper cites at most two others. Moreover,~$D$ is acyclic with topological ordering $(\overrightarrow{Z}, \overrightarrow{X},  v^{\star}, \overrightarrow{U}, c, y)$. 

   \item[Papers allowed to be merged.] Let $W=U\cup\{v^{\star}, c\}$.
   
    \item[Compatibility graph~$H$.] The compatibility graph~$H$ contains the following edges: 
    \begin{itemize}
        \item For each $v\in V$, all edges between papers in $P(v)$.
        \item For each $v, u\in V$ and each $i\in [q-p]$, there is an edge between $v(i)$ and $u(i)$ if and only if $v$ and $u$ are adjacent in $G$.
        \item All edges between $U$ and $c$.
    \end{itemize}
    \end{description}
   
    We now prove the correctness.

    $(\Rightarrow)$ Assume that $G$ admits a $\kappa$-clique $K$. Let $A=\bigcup_{v\in K}P(v)$. By construction, $A\cup \{c\}$ forms a clique in $H$. After merging  papers in $A\cup \{c\}$ into a single paper, the {\cd} index of~$v^{\star}$ becomes \[\frac{\abs{X}}{\abs{X}+\abs{U\setminus A}}=\frac{p\cdot (n-\kappa)}{p\cdot (n-\kappa)+(n\cdot (q-p)-\kappa\cdot (q-p))}=\frac{p}{q}.\] Thus,~$I$ is a {\yesins} of {\mergemanipulation}.

    $(\Leftarrow)$ Assume that~$I$ is a {\yesins}. Then, there is a partition~$\mathcal{P}$ of~$W$ that complies with~$H$ and at least one paper from~$W^{\mathcal{P}}$ has a {\cd} index of at least~${p}/{q}$ in~$D^{\mathcal{P}}$. Let $P\in \mathcal{P}$ be the part containing~$c$. Since $c$ is adjacent to all papers in $U$ in the compatibility graph $H$, we may assume that $\abs{P}\ge 2$ whenever merging is beneficial.  Moreover, since $\beta=1$,  every other part of $\mathcal{P}$ is a singleton. 
    The presence of $y$, which is cited by all papers in $U\cup \{c\}$, together with the papers in $Z$ citing $y$, ensures that no paper in $(U\cup \{c\})^{\mathcal{P}}$ can attain a {\cd}-index of at least $d$. Consequently,~$v^{\star}$ is the only paper that can have a {\cd} index of at least~$d$ in~$D^{\mathcal{P}}$. 
    Note that $\cramped{\citeF{v^{\star}}{D^{\mathcal{P}}}}=X$ and $\citeB{v^{\star}}{D^{\mathcal{P}}}=\emptyset'$. 
    Let $P'=P\setminus \{c\}$. Then, the {\cd} index of~$v^{\star}$ in~$D^{\cramped{\mathcal{P}}}$ is 
    $[\frac{\abs{X}}{\abs{X}+(\abs{U}- \abs{P'}}$. 
    This value is at least $d=p/q$ only if $\abs{P'}\geq \kappa\cdot (q-p)$. This implies that there exists some $i\in [q-p]$ such that the set $\{v(i)\setmid v\in V\}$ contains at least $\kappa$ papers. These papers can belong to the same part $P$ only if they form a clique in $H$. By construction, the corresponding vertices in $G$ forms a clique. Hence, $G$ admits a $\kappa$-clique.   
\end{proof}

This reduction indicates that a core challenge for the single-target case lies within the structure of the compatibility graph. This prompts us to explore restricted compatibility graphs where each connected component contains only a constant number of vertices. Remarkably, for $d=1$, the problem becomes tractable even when multiple papers are targeted, provided that their number is bounded by a constant under this restriction.  
The core idea is to first guess (enumerate) which connected components (at most $\ell$) will contain the papers that eventually attain a CD index of $1$. We then enumerate all possible partitions of the papers within these components and solve the remaining instances greedily.

\begin{theorem}
\label{thm-merge-poly-constant-compatibility-ell-const}
    For $d=1$, {\mergemanipulation} is in {\xp} {\wrt}~$\ell$ when each connected component of $H[W]$ contains a constant number of papers, where $H$ is the compatibility graph and $W$ is the set of papers eligible for merging.
\end{theorem}

\begin{proof}
      Let $I=(D, W, H, \ell, \ss, d)$ be an instance of  {\mergemanipulation} with $d=1$. Let $\textsf{CC}$ denote the set of connected components of~$H[W]$.
      . We enumerate all subsets~$S$ of $\textsf{CC}$ with $\abs{S}\leq \ell$, i.e., $S$ contains at most~$\ell$ connected components of~$H$. For each such~$S$, let $\vset{S}=\bigcup_{H'\in S}\vset{H'}$ be the set of vertices of~$H$ contained in at least one element of~$S$. We then enumerate all partitions~$\mathcal{P}$ of~$\vset{S}$ that comply with~$H$ and satisfy $\abs{S^{\mathcal{P}}}\geq \ell$. For each such~$\mathcal{P}$, we enumerate all~$\ell$-subsets $\hat{S}\subseteq S^{\mathcal{P}}$. Papers in~$\hat{S}$ are required to have a {\cd} index of~$1$ in the final citation digraph. 
    Next, the algorithm handles the connected components of~$H[W]$ not in~$S$ one by one in arbitrary order. Let~$H'$ denote the currently considered  component. For every compound paper $v\in \hat{S}$, let~$R(v)$ be the set of papers in~$H'$ cited by at least one atom paper of~$v$, and let $C(v)$ be the set of papers in~$\vset{H'}\setminus R(v)$ citing at least one paper in~$R(v)$. 
    For every $v\in \hat{S}$ to attain a {\cd} index of~$1$ through merging papers in~$W$ in a manner that complies with the compatibility graph~$H$, it is necessary that each paper in~$C(v)$ is merged with at least one paper from~$R(v)$. 
    Thus the algorithm enumerates all partitions~$\mathcal{P'}$ of~$\vset{H'}$ that comply with~$H$ and satisfy the following condition:
    \begin{itemize}
        \item for every $v\in \hat{S}$ with $C(v)\neq\emptyset$, each paper~$u\in C(v)$ is contained in a part together with at least one paper from~$R(v)$, i.e., there exists $P\in \mathcal{P'}$ such that $u\in P$ and $P\cap R(v)\neq\emptyset$. 
    \end{itemize}
It may happen that no such partition exists. In this case, we discard the current choice of~$\hat{S}$. Otherwise, among all such partitions of $V(H')$, we arbitrarily select one $\mathcal{P}'$ that minimizes $|\mathcal{P}'^{\geq 2}|$ and set $\mathcal{P} \coloneqq \mathcal{P} \cup \mathcal{P}'$. 
If, after processing all connected components in~$\textsf{CC}\setminus S$, we have $\abs{\mathcal{P}^{\geq 2}}\leq \ss$ and all papers in~$\hat{S}$ have a {\cd} index of~$1$ in~$D^{\mathcal{P}}$, the algorithm returns ``\yes''. 
%
If the algorithm exhausts all enumerations without returning ``\yes'', we conclude that~$I$ is a {\noins}. 
The pseudocode is given in Algorithm~\ref{alg}. 

To see that the algorithm runs in polynomial time, let $n$ denote the number of papers in $D$, which is an upper bound on $\lvert \textsf{CC} \rvert$. Let $c$ denote the maximum number of papers in any connected component of $H[W]$. The loop in Line~2 enumerates $\bigo{n^{\ell}}$ subsets~$S$, the loop in Line~4 enumerates $\bigo{(c\cdot \ell)^{c\cdot \ell}}$ partitions~$\mathcal{P}$, the loop in Line~5 considers at most $2^{c\cdot \ell}$ choices for~$\hat{S}$, and the loop in Line~6 considers at most $n$ graphs~$H'$. Moreover, computing a partition~$\mathcal{P}'$ in Line~11 can be done in time $\bigo{c^c}$ by enumerating all partitions of~$\vset{H'}$ in Line~10. Therefore, the algorithm runs in polynomial time when $\ell$ and $c$ are constants.
\end{proof}

\begin{algorithm*}
\SetKwData{Left}{left}\SetKwData{This}{this}\SetKwData{Up}{up}
\SetKwFunction{Union}{Union}\SetKwFunction{FindCompress}{FindCompress}
\SetKwInOut{Input}{Input}\SetKwInOut{Output}{Output}
\Input{An instance $I=(D, W, H, \ell, \ss, d)$ of {\mergemanipulation} with $d=1$.}
\Output{``{\yes}'' if $I$ is a {\yesins}, and``\no'' otherwise.}
let $\textsf{CC}$ be the set of connected components of $H[W]$\;
\For{all subsets $S$ of $\textsf{CC}$ with $\abs{S}\leq \ell$}
{
let $\vset{S}=\bigcup_{H'\in S}\vset{H'}$\;
  \For{all partitions $\mathcal{P}$ of $\vset{S}$ that comply with~$H$ and satisfy $\abs{S^{\mathcal{P}}}\geq \ell$}
  {
     \For{all $\ell$-subsets $\hat{S}$ of $S^{\mathcal{P}}$}
     {
        \For{all $H'\in \textsf{CC}\setminus S$}
        {
           \For{all $v\in \hat{S}$}{let $R(v)=\outneighbor{v}{D}\cap \vset{H'}$\; let $C(v)=\{u\in \vset{H'}\setminus R(v)\setmid \outneighbor{u}{D}\cap R(v)\neq\emptyset\}$\;}
           \eIf{there exists a partition $\mathcal{P}'$ of $\vset{H'}$ that complies with $H$ such that for every $v\in \hat{S}$ with $C(v)\neq\emptyset$, it holds that for every~$u\in C(v)$ there exists $P\in \mathcal{P}'$ containing~$u$ with $P\cap R(v)\neq \emptyset$}{let $\mathcal{P'}$ be such a partition with the minimum $\abs{\mathcal{P}'^{\geq 2}}$\;
           $\mathcal{P}\coloneqq \mathcal{P}\cup \mathcal{P}'$\;
           }{break\;}
        }
        \If{$\abs{\mathcal{P}^{\geq 2}}\leq \ss$ and every paper in $\hat{S}$ has {\cd} index~$1$ in~$D^{\mathcal{P}}$}{return ``\yes''\;}
     }
  }
}
  return ``\no''\;
\caption{An algorithm for solving {\mergemanipulation} with $d=1$}
\label{alg}
\end{algorithm*}

Treating~$\ss$ as the parameter, we obtain an {\xp} algorithm for all possible~$d$. The main idea is to guess which connected components of the compatibility graph have multiple papers being merged, and then use a brute-force approach to solve the problem. 

\begin{theorem}
    \label{thm-merge-xp-ss-small-acycli}
    {\mergemanipulation} is in {\xp} with respect to~$\ss$, when each connected component of $H[W]$ contains a constant number of papers, where $H$ is the compatibility graph and $W$ is the set of papers eligible for merging.
\end{theorem}

\begin{proof}
    Let $I=(D, W, H, \ell, \ss, d)$ be an instance of  {\mergemanipulation}. Let $\textsf{CC}$ be the set of connected components of~$H[W]$. We enumerate all subsets~$S$ of $\textsf{CC}$ with $\abs{S}\leq \ss$. Here,~$S$ is supposed to be exactly the set of connected components in~$H[W]$ where at least one nonsingleton set of papers is merged into a single entity. Hence, we impose the condition $\abs{S} \leq \ss$. Papers outside of~$S$'s connected components are not intended for merging. For each enumerated~$S$, let $\vset{S}=\bigcup_{H'\in S}\vset{H'}$. 
    Next, we enumerate all partitions~$\mathcal{P}$ of~$\vset{S}$. If a partition~$\mathcal{P}$ complies with~$H$, satisfies $\abs{\mathcal{P}^{\geq 2}}\leq \ss$, and ensures that at least~$\ell$ papers in ${\vset{S}^{\mathcal{P}}}\cup (W\setminus \vset{S})$ have {\cd} indices of at least~$d$ in~$D^{\mathcal{P}}$, we return ``\yes''. 
    If no such partitions exist after enumerating all possibilities, we return~``\no''.

     Let $c$ denote the maximum number of papers in any connected component of $H[W]$. Since there are at most $n^{\ss}$ subsets~$S$ to enumerate, and at most $(\ss\cdot c)^{\ss\cdot c}$ partitions  for each~$S$, where $n=\abs{W}$, the algorithm runs in polynomial time, provided~$\ss$ and~$c$ are constant.
\end{proof}

We now consider the variant {\mergemanipulationinf}, where the number of paper merges is unrestricted. The reduction used in the proof of Theorem~\ref{thm-merge-d-1-np-hard} also applies here, yielding the following corollary.

\begin{corollary}
    \label{cor-merge-infty-d-1-np-hard}
    For any positive rational number $d\leq 1$, {\mergemanipulationinf} is {\nph}, even if $\ell=1$ and the citation digraph is acyclic.
\end{corollary}

Unfortunately, for $d=1$, the complexity of the previously studied special cases remains unresolved.

\begin{openquestion}
What is the computational complexity of {\mergemanipulationinf} for $d=1$ under the following restrictions: (i) the citation digraph is acyclic, (ii) each connected component of the compatibility graph contains only a constant number of papers, or (iii) both restrictions hold?
\end{openquestion}

Nevertheless, since {\mergemanipulationinf} is polynomial-time Turing reducible to {\mergemanipulation}, Theorem~\ref{thm-merge-poly-constant-compatibility-ell-const} immediately yields the following {\xp} upper bound.

\begin{corollary}
\label{cor-merge-poly-constant-compatibility-ell-const}
For $d=1$, {\mergemanipulationinf} is in {\xp} {\wrt}~$\ell$ when each connected component of the compatibility graph induced by the set $W$ of mergeable papers contains a constant number of papers.
\end{corollary}

To conclude this section, we observe that {\mergemanipulation} can be solved by enumerating all possible partitions of~$W$, implying that it is {\fpt} with respect to~$\abs{W}$, the number of papers eligible for merging. Furthermore, since {\mergemanipulationinf} is polynomial-time Turing-reducible to {\mergemanipulation}, we obtain the following result.

\begin{corollary}
   {\mergemanipulationinf} is {\fpt} with respect to ${\abs{W}}$, when each connected component of the compatibility graph contains a constant number of papers.
\end{corollary}

\section{Manipulation by Paper Addition}
\label{sec-addition}
Now we examine the operation of paper addition.
We observe that {\addimprovement} is in {\xp} when parameterized by~$k$: it can be solved by enumerating all possible sets of at most~$k$ papers to be added.
This raises the question of whether we can enhance the result to fixed-parameter tractability. The next theorem answers this question in the negative.

\begin{theorem}
\label{thm-add-wbh-k-all-d}
For any positive rational number $d\leq 1$, {\addimprovement} is {\wbh} when parameterized by~$k$, 
even if the citation digraph is acyclic and bipartite.
\end{theorem}

\begin{proof}
    We prove the theorem by a reduction from {\prob{RBDS}}. Let $I=(G,\kappa)$ be an instance of {\prob{RBDS}}, where~$G$ is a bipartite graph with the bipartition $(B, R)$. Construct a citation digraph $D$ by orienting each edge $\edge{b}{r}$ with $b\in B$ and $r\in R$ as an arc $(b,r)$. Clearly, $D$ is acyclic and bipartite. Let $V=J=R$, $U=B$, and $k=\kappa$. The instance of {\addimprovement} is $g(I)=(D, (V\cup U), J, k, d)$. 

If there is a subset $B'\subseteq B$ of at most $\kappa$ vertices that dominate~$R$ in~$G$, then in the citation digraph $D[V\cup B']$, every paper $r\in R$ receives at least one citation from~$B'$ and has {\cd} index~$1$. Since $k=\kappa$ and $d\leq 1$,~$g(I)$ is a {\yesins}.

For the converse, assume that there exists $B'\subseteq B$ with $\abs{B'}\leq k$ such that the {\cd} index of every paper of~$R$ in $D[V\cup B']$ is at least~$d$. This implies that every paper $r\in R$ receives at least one citation from~$B'$ in~$D[V\cup B']$. By construction, this indicates that~$B'$ dominates~$R$ in~$G$. Since $k=\kappa$,~{$I$} is a {\yesins}.
\end{proof}


Furthermore, we show that restricting the value of~$d$ to~$1$ does not render the complexity of the problem {\fpt}, even when associated with a larger parameter.

\begin{theorem}
\label{thm-add-wah-j-k}
For $d=1$, {\addimprovement} is {\wah} when parameterized by~$\abs{J}+k$, even if the citation digraph is acyclic.
\end{theorem}

\begin{proof}
    We prove the theorem via a reduction from {\prob{Multicolored Clique}}. Let $I=(G, (V_1, V_2, \dots, V_{\kappa}))$ be an instance of {\prob{Multicolored Clique}} with $\kappa \geq 2$. We create an instance $g(I)=(D, (V, U), J, k, d)$ with $d=1$ of {\addimprovement} as follows (see Figure~\ref{fig-add-wah-j-k} for an illustration). 
\begin{description}
    \item[Citation digraph~$D$.] For each $v\in \vset{G}$, we create one paper denoted by the same symbol for simplicity. For each edge $e=\edge{v}{u}$ in~$G$, we create an edge-paper~$p(e)$. For each pair $\{i, j\}\subseteq [\kappa]$, let $P_{\{i,j\}}=\{p(e) \setmid e\in \eset{G}, e\cap V_i\neq \emptyset, e\cap V_j\neq\emptyset\}$. 
    In addition, we create a set $X=\{x_1, x_2, \dots, x_{\kappa}\}$ of~$\kappa$ papers, and a set $Y=\{y_{\{i,j\}}\}_{\{i,j\}\subseteq [\kappa]}$ of ${\kappa \cdot (\kappa-1)}/{2}$ papers. The arcs of~$D$ are defined as follows: 
    \begin{itemize}
        \item For every $i\in [\kappa]$, all papers from~$V_i$ cite~$x_i$.
        \item For every pair $\{i, j\}\subseteq [\kappa]$, all papers from~$P_{\{i,j\}}$ cite~$y_{\{i,j\}}$.
        \item For every pair $\{i,j\}\subseteq [\kappa]$ and every edge $e=\edge{v}{u}$ between~$V_i$ and~$V_j$ in~$G$, the paper~$p(e)$ cites all papers from $(V_i\cup V_j)\setminus \{v, u\}$.
        \item For every pair $\{i,j\}\subseteq [\kappa]$, the paper~$y_{\{i,j\}}$ cites all papers from~$V_i\cup V_j$.
    \end{itemize}
\end{description}
Let $V=J=X\cup Y$, $P(E(G)) = \bigcup_{\{i,j\} \subseteq [\kappa]} P_{\{i,j\}}$, $k=\kappa+\frac{\kappa \cdot (\kappa-1)}{2}$, and define $U = P(E(G)) \cup \vset{G}$. 
We have $\abs{J}+k=2\kappa+\kappa\cdot(\kappa-1)$. 
    The citation digraph~$D$ is acyclic with the topological ordering $(\overrightarrow{P(E(G))}, Y, \overrightarrow{\vset{G}}, \overrightarrow{X})$.
The construction of~$g(I)$ runs in polynomial time. We now prove the correctness of the reduction.
\begin{figure}
    \centering
    \includegraphics[width=0.5\textwidth]{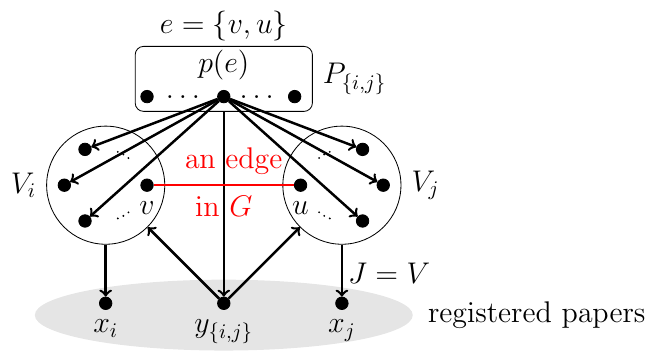}
    \caption{Illustration of the reduction in the proof of Theorem~\ref{thm-add-wah-j-k}.}
    \label{fig-add-wah-j-k}
\end{figure}

    $(\Rightarrow)$ Assume that there is a~$\kappa$-colored clique~$K$ in~$G$. Let $P(K)=\{p(e) \setmid e\in \eset{G[K]}\}$ be the set of the ${\kappa\cdot (\kappa-1)}/{2}$ papers corresponding to the edges within the clique~$K$. Consider the citation digraph~$D'=D[V\cup K\cup P(K)]$. For every $x_i\in X$, as $\abs{K\cap V_i}=1$, the construction of the citation digraph~$D$ ensures that the {\cd} index of~$x_i$ in~$D'$ is~$1$. Now consider a paper~$y_{\{i,j\}}$, where $\{i,j\}\subseteq [\kappa]$. 
    We have $\abs{P(K)\cap P_{\{i,j\}}}=1$. 
    Let $P(K)\cap P_{\{i,j\}}=\{p(e)\}$, where $e=\edge{v}{u}$ is an edge in ${G}$ with $v\in V_i$ and $u\in V_j$. Consequently, $K\cap V_i=\{v\}$ and $K\cap V_j=\{u\}$. According to the construction of~$D$, in the citation digraph~$D'$,~$v$ and~$u$ are the only references of~$y_{\{i,j\}}$,~$p(e)$ is the only paper citing~$y_{\{i,j\}}$, and~$p(e)$ cites neither~$v$ nor~$u$. As a result, the {\cd} index of~$y_{\{i,j\}}$ in~$D'$ is~$1$. Therefore, all papers in $J=X\cup Y$ have a {\cd} index of~$1$ in~$D'$. Since $\abs{K\cup P(K)}=\kappa+\frac{\kappa\cdot(\kappa-1)}{2}=k$,~$g(I)$ is a {\yesins}.

    $(\Leftarrow)$ Assume that there exists $U'\subseteq U$ with $\abs{U'}\leq k$ such that every paper in~$J$ has a {\cd} index of~$1$ in $D'=D[V\cup U']$. A necessary condition for a paper $x_i\in X$ to have a {\cd} index of~$1$ in~$D'$ is that~$U'$ contains at least one paper from~$V_i$. A necessary condition for a paper $y_{\{i,j\}}\in Y$ to have a {\cd} index of~$1$ in~$D'$ is that~$U'$ contains at least one paper from~$P_{\{i,j\}}$. Since $k=\kappa+\frac{\kappa \cdot (\kappa-1)}{2}$, we know that $\abs{U'\cap V_i}=1$ for all $i\in[\kappa]$, and $\abs{U'\cap P_{\{i,j\}}}=1$ for all pairs $\{i,j\}\subseteq [\kappa]$. 
    Let~$K=U'\cap \vset{G}$. 
    We prove that~$K$ is a clique in~$G$. Assume for contradiction that~$K$ contains two vertices~$v$ and~$u$ that are nonadjacent in~$G$. Let $v\in V_i$ and $u\in V_j$ for some pair $\{i,j\}\subseteq [\kappa]$. Let $\{p(e)\}=U'\cap P_{\{i,j\}}$ where~$e$ is an edge in~$G$ crossing~$V_i$ and~$V_j$. Clearly, $e\neq \{v,u\}$. Then, by construction,~$p(e)$ cites at least one of~$v$ and~$u$. However, as $y_{\{i,j\}}$ cites both~$v$ and~$u$,  the {\cd} index of~$y_{\{i,j\}}$ in~$D'$ cannot be~$1$, a contradiction. This completes the proof that~$K$ is a clique in~$G$. Then, from $\abs{U'\cap V_i}=1$ for all $i\in [\kappa]$, it follows that~$K$ is a $\kappa$-colored clique in~$G$.
\end{proof}


\onlyfull{
\begin{theorem}
\label{thm-d-smaller-1-add-wah-k-j}
    For any positive rational number $d<1$, {\addimprovement} is {\wah} when parameterized by~$\abs{J}+k$\onlyfull{, the number of papers allowed to be added plus the number of focal papers}. This holds even if the citation digraph is acyclic.
\end{theorem}

\begin{proof}
Let $\frac{p}{q}$ be the canonical form of $d$. We make $p$. We construct a citation digraph $\mathcal{D}$ which contains $p$ disjoint copies of the citation digraph $D$ constructed in the proof of Theorem~\ref{thm-add-wah-j-k}, together with a new paper $z$, and a set $Z$ of $(q-p)\cdot (\kappa+\frac{\kappa \cdot (\kappa-1)}{2})$ new papers. So, the citation digraph $\mathcal{D}$ contains $p\cdot \abs{\vset{D}}+1+(q-p)\cdot (\kappa+\frac{\kappa \cdot (\kappa-1)}{2})$ papers.
\end{proof}
}

The above reduction without~$k$ directly applies to {\addimprovementinf}. 

\begin{theorem}
\label{thm-add-inf-wah-j-d-1}
For $d=1$, {\addimprovementinf} is {\wah} when parameterized by~$\abs{J}$\onlyfull{, the number of papers allowed to be added plus the number of focal papers}. This holds even if the citation digraph is acyclic.
\end{theorem}

\begin{proof}
We use the same reduction as in the proof of Theorem~\ref{thm-add-wah-j-k}, omitting the parameter $k$.  
We have already shown that if $G$ contains a $\kappa$-colored clique $K$, then after adding all papers in $K \cup P(K)$, every paper in $J$ has {\cd} index $1$. For the reverse direction, assume that there exists $U' \subseteq U$ such that every paper in~$J$ has {\cd} index $1$ in $D' = D[V \cup U']$. 
As argued earlier, it holds that $\abs{U' \cap V_i} \geq 1$ for all $i \in [\kappa]$, and $\abs{U' \cap P_{\{i,j\}}} \geq 1$ for all $\{i,j\} \subseteq [\kappa]$. A crucial observation is that in fact $\abs{U' \cap V_i} = 1$ and $\abs{U' \cap P_{\{i,j\}}} = 1$ must hold. Suppose, for contradiction, that this is not the case. Then one of the following two cases occurs.
\begin{description}
    \item[Case~1:] There exists $i \in [\kappa]$ such that $\abs{U' \cap V_i} \geq 2$. \hfill 
    
    Let $v, v'$ be two distinct papers in $U' \cap V_i$, fix an arbitrary $j \in [\kappa] \setminus \{i\}$, and let $p(e)$ (for some $e \in \eset{G}$) be an arbitrary paper in $U' \cap P_{\{i,j\}}$. By construction, $p(e)$ cites $y_{\{i,j\}}$ and at least one of $v$ and $v'$, while $y_{\{i,j\}}$ cites both $v$ and $v'$. This implies that the {\cd} index of $y_{\{i,j\}}$ in $D'$ is strictly smaller than $1$, a contradiction.

   \item[Case 2:] There exists $\{i,j\} \subseteq [\kappa]$ such that $\abs{U' \cap P_{\{i,j\}}} \geq 2$. \hfill 
   
   Let $v$ be an arbitrary paper in $U' \cap V_i$, let $u$ be an arbitrary paper in $U'\cap V_j$, and let $p(e)$ and $p(e')$ be two distinct papers in $U' \cap P_{\{i,j\}}$. By construction, at least one of $p(e)$ and $p(e')$ cites at least one of $v$ and $u$. Moreover, both~$p(e)$ and~$p(e')$ cite $y_{\{i,j\}}$, and $y_{\{i,j\}}$ cites both $v$ and $u$. Hence, the {\cd} index of $y_{\{i,j\}}$ in $D'$ must be strictly smaller than $1$, again a contradiction.
\end{description}
Therefore, we conclude that $\abs{U' \cap V_i} = 1$ for all $i \in [\kappa]$ and $\abs{U' \cap P_{\{i,j\}}} = 1$ for all $\{i,j\} \subseteq [\kappa]$. It follows that $\abs{U'} = \kappa + \frac{\kappa(\kappa - 1)}{2}$. 
The remainder of the argument is identical to that in the proof of Theorem~\ref{thm-add-wah-j-k}.
\end{proof}

Theorem~\ref{thm-add-wah-j-k} dismisses the feasibility of achieving {\fpt} algorithms when parameterized by~$\abs{J}$, unless {\fpt$=$\wa}. Nevertheless, there still exists a potential for the problem to become tractable when~$\abs{J}$ is a constant. We first explore the case where~$J$ is a singleton, and show that it can be solved in polynomial time by a greedy algorithm.

\begin{theorem}
\label{them-add-p-j-1}
    {\addimprovement} is polynomial-time solvable when $\abs{J}=1$.
\end{theorem}

\begin{proof}
    Let $(D, (V, U), J, k, d)$ be an instance of {\addimprovement}, where $J=\{v^{\star}\}$. 
    Let \[U'=\{u\in U\cap \inneighbor{v^{\star}}{D}\setminus \outneighbor{v^{\star}}{D} \setmid V\cap \outneighbor{v^{\star}}{D}\cap \outneighbor{u}{D}= \emptyset\}\] 
    be the set of unregistered non-reciprocal citations of $v^{\star}$ that do not cite any registered references of~$v^{\star}$. 
    Let $k'=\min\{k,\abs{U'}\}$, and let~$U''$ be any arbitrary $k'$-subset of~$U'$. We return ``\yes'' if and only if the {\cd} index of~$v^{\star}$ in $D[V\cup U'']$ is at least~$d$. The algorithm runs in polynomial time, and its correctness is straightforward. 
\end{proof}

Since {\addimprovementinf} is polynomial-time Turing reducible to {\addimprovement}, we have the following result. 

{
\begin{corollary}
\label{cor-add-inf-p-j-1}
    {\addimprovementinf} is polynomial-time solvable when $\abs{J}=1$.
\end{corollary}
}

We next show that for $d=1$, tractability extends to all constant-bounded~$\abs{J}$. 

\begin{theorem}
\label{thm-d-1-add-xp-j}
    For $d=1$,  {\addimprovement} and {\addimprovementinf} are in {\xp} parameterized by~${\abs{J}}$.
\end{theorem}

\begin{proof}
 Let $I=(D, (V, U), J, k, d)$ be an instance of {\addimprovement}, where $d=1$. 
    Let \[J'=\{p\in J\setmid (\inneighbor{p}{D}\setminus \outneighbor{p}{D})\cap V=\emptyset\}\] be the set of papers in $J$ without any registered non-reciprocal citations. The algorithm enumerates all functions $f: J' \to U$ such that $f(p)\in \inneighbor{p}{D}$ for all $p\in J'$. 
If there exists such a function $f$ with $|\operatorname{Im}(f)|\le k$ and such that every $p\in J$ has {\cd} index $1$ in $D[V\cup \operatorname{Im}(f)]$, then return ``\yes'', where $\operatorname{Im}(f)=\{f(p)\mid p\in J'\}$ denotes the image of~$f$. 
Otherwise, return ``\no''.

The algorithm considers at most $\abs{U}^{\abs{J'}} \le \abs{U}^{\abs{J}}$ functions~$f$, and each verification can be done in polynomial time. Hence, the algorithm runs in {\xp} time with respect to~$\abs{J}$. Correctness follows directly from the construction. 
Since {\addimprovementinf} is polynomial-time Turing reducible to {\addimprovement}, it is also in {\xp} with respect to~$\abs{J}$.
 \end{proof}

\onlyfull{
The algorithm outlined above can be readily adapted to solve the variant {\addimprovementinf} for $d=1$. In particular, since we do not care about the number of papers added, in Step~2, we remove the premise condition $\abs{U'}>k$. Thus, we  discard an enumerated vector~$\abs{U'}$ only if there exists $p\in P$ such that the {\cd} index of~$p$ in $D[V\cup U']$ is not~$1$.

\begin{corollary}
\label{cor-d-1-add-inf-xp-j}
      For $d=1$,  {\addimprovementinf} can be solved in polynomial time if $\abs{J}$ is bounded from above by  a constant.
\end{corollary}
}
For other positive values of~$d$, the parameter $\abs{J}$ complicates the problem, as stated in the following  theorem.

\begin{theorem}
\label{thm-add-wah-k-j-2-d-smaller}
For any positive rational number $d<1$,  {\addimprovement} is {\wah} with respect to~$k$, even if $\abs{J}=2$ and the citation digraph is acyclic and bipartite.
\end{theorem}

\begin{proof}
    We prove the theorem via a reduction from {\prob{Clique}} on regular graphs. Let $I=(G,\kappa)$ be an instance of {\prob{Clique}}, where~$G$ is $t$-regular. Without loss of generality, assume $1\leq \kappa<t<\abs{\vset{G}}$. We first present the reduction for $d=1/2$ and then extend it to any positive rational $d<1$. 
Let $m_{\text{c}}=\kappa\cdot t-\frac{\kappa \cdot (\kappa-1)}{2}$. This is the number of edges covered by a $\kappa$-clique in a $t$-regular graph, and it is also the minimum number of edges covered by any~$\kappa$ vertices. 
\begin{description}
    \item[Citation digraph~$D$.] The set of papers in $D$ is $\vset{G}\cup \eset{G}\cup X\cup Y\cup Z\cup \{v^{\star}_1, v^{\star}_2, \Tilde{v}\}$, where~$\abs{X}=\abs{Y}=m_{\text{c}}$ and~$\abs{Z}=m_{\text{c}}-\kappa$. Let $J=\{v^{\star}_1, v^{\star}_2\}$, let $U=\vset{G}$ be the set of unregistered papers, and let~$V=\vset{D}\setminus U$ be the set of registered papers.
    The arcs in $D$ are defined by the following citations: 
    \begin{itemize}
        \item Every edge-paper $e=\edge{v}{u}\in \eset{G}$ cites the two papers~$v$ and~$u$. 
        \item All papers in $\vset{G}\cup Z$ cite $v^{\star}_1$.
        \item The paper~$v^{\star}_2$ cites all papers from $\vset{G}$.
        \item All papers in~$Y$ cite $v^{\star}_2$.
        \item All papers in~$X\cup \{v^{\star}_1\}$ cite~$\Tilde{v}$.
    \end{itemize}
\end{description}
    See Figure~\ref{fig-add-wah-k-j-2-d-smaller} for an illustration. It is evident that~$D$ is bipartite. Moreover, it is acyclic with the topological ordering $(\overrightarrow{E(G)}, \overrightarrow{Y}, v^{\star}_2, \overrightarrow{V(G)}, \overrightarrow{X}, \overrightarrow{Z}, v^{\star}_1, \tilde{v})$.
    \begin{figure}
        \centering
        \includegraphics[width=0.35\textwidth]{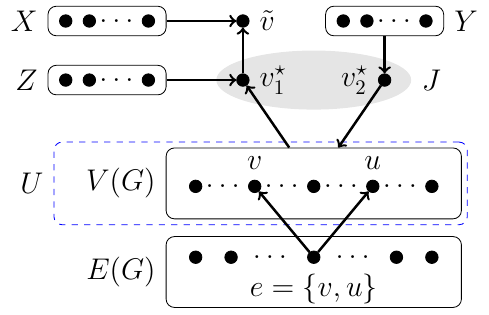}
        \caption{Illustration of the reduction in the proof of Theorem~\ref{thm-add-wah-k-j-2-d-smaller}. The set of unregistered papers is $U=\vset{G}$.}
        \label{fig-add-wah-k-j-2-d-smaller}
    \end{figure}
    To complete the reduction, we set the budget $k = \kappa$ and the target index $d = 1/2$. The resulting instance of {\addimprovement} is $g(I) = (D, (V, U), J, k, d)$. 

    $(\Rightarrow)$ Suppose that~$G$ has a $\kappa$-clique~$K$. The {\cd} index of~$v^{\star}_1$ in~$D[V\cup K]$ is
    $\frac{\abs{K}+\abs{Z}}{\abs{K}+\abs{Z}+\abs{X}}=d$. We now consider~$v^{\star}_2$. Recall that
    $\spane{K}{G}=\{e\in \eset{G} \setmid e\cap K\neq \emptyset\}$ is the set of edges in~$G$ covered by~$K$. 
    We have $\abs{\spane{K}{G}}=m_{\text{c}}$. By construction, the {\cd} index  of~$v^{\star}_2$ in~$D[V\cup K]$ is
    \begin{equation}
    \label{eq-cd-vstar-2}
        \frac{\abs{Y}}{\abs{Y}+\abs{\spane{K}{G}}}=d.
    \end{equation} Since $\abs{K}=\kappa=k$, it follows that~$g(I)$ is a {\yesins}.

    $(\Leftarrow)$ Assume that there exists $K\subseteq U$ such that $\abs{K}\leq k=\kappa$, and the {\cd} indices of both $v^{\star}_1$ and $v^{\star}_2$ in $D'=D[V\cup K]$ are at least~$1/2$. The {\cd} index of~$v^{\star}_2$ in~$D'$ is given by~\eqref{eq-cd-vstar-2}. Hence, $\abs{\spane{K}{G}}\leq m_{\text{c}}$. Since $\abs{K}\leq \kappa$ and~$G$ is $t$-regular, it follows that~$K$ is a clique in~$G$. Thus~$I$ is a {\yesins}.

   We extend the reduction to any positive rational $d<1$. Let $d=p/q$ be in canonical form. The construction is modified as follows:
\begin{itemize}
    \item Add $(q-p-1)\cdot m_{\text{c}}$ papers to~$X$, all citing $\Tilde{v}$, $(p-1)\cdot m_{\text{c}}$ papers to~$Y$, all citing $v^{\star}_2$, and $(p-1)\cdot m_{\text{c}}$ papers to~$Z$, all citing $v^{\star}_1$. Then $\abs{X}=(q-p)\cdot m_{\text{c}}$, $\abs{Y}=p\cdot m_{\text{c}}$, and $\abs{Z}=p\cdot m_{\text{c}}-\kappa$.
    \item Add $(q-p-1)$ copies of $\eset{G}$, where each copy of a paper $e=\edge{v}{u}$ cites $v$ and $u$.
\end{itemize}
    The correctness argument is analogous. 
\end{proof}

The same reduction applies to {\addimprovementinf}.

\begin{theorem}
    \label{cor-add-inf-np-hard-j-2-d-smaller}
For any positive rational number $d<1$,  {\addimprovementinf} is {\nph}, even if $\abs{J}=2$ and the citation digraph is acyclic and bipartite.
\end{theorem}

\begin{proof}
We use the reduction from the proof of Theorem~\ref{thm-add-wah-k-j-2-d-smaller}, omitting~$k$ from the instance. 
The $(\Rightarrow)$-direction remains unchanged. For the $(\Leftarrow)$-direction, observe that adding more than $\kappa$ papers from $U$ to $V$ causes the {\cd} index of $v^{\star}_2$ to drop strictly below the threshold $d$.
\end{proof}

\section{Manipulation by Paper Deletion}
\label{sec-deletion}
We now consider the problem of {\cd} index manipulation by deleting papers. 
The problem is in {\xp} with respect to~$k$, the number of papers allowed to be deleted. This raises the question of whether it is {\fpt} for the same parameter. In contrast to the intractability result for {\cd} index manipulation by adding papers, we show that the complexity of {\deleteimprovement} depends on~$d$: it is {\fpt} with respect to~$k$ when $d=1$, but {\wah} for every rational $d<1$. 
We use $\bigos{\cdot}$ to denote $\bigo{\cdot}$ with polynomial factors suppressed.

\begin{theorem}
\label{thm-del-d-1-fpt-k}
    For $d=1$, {\deleteimprovement}  can be solved in time $\bigos{2^k}$.
\end{theorem}

\begin{proof}
Let $I=(D, J, k, d)$ be an instance of {\deleteimprovement} with $d=1$. 
%
If there exists a paper $v\in J$ such that $\citeB{v}{D[J]}\cup \citeR{v}{D[J]}\neq\emptyset$, we directly return ``{\no}''. 
Henceforth, we assume that $\citeB{v}{D[J]}\cup \citeR{v}{D[J]}=\emptyset$ holds for all $v\in J$. 
We iteratively apply the following two reduction rules until neither is applicable. An illustration of the rules is shown in Figure~\ref{reduction-rule-del}.  

\begin{figure}
    \centering
    \includegraphics[width=0.75\linewidth]{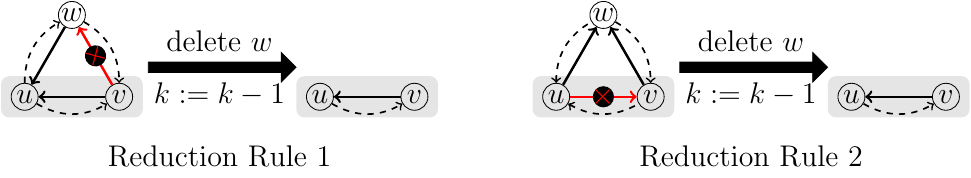}
    \caption{Illustration of Reduction Rules~\ref{reduction-rule-a} and~\ref{reduction-rule-b}. Dashed arcs indicate optional connections (they may or may not exist), while arcs marked with a red cross are absent. The two vertices within the shaded regions belong to $J$, the set of distinguished papers; the remaining vertex lies outside $J$.}
    \label{reduction-rule-del}
\end{figure}

\begin{reductionrule}
\label{reduction-rule-a}
    If there exist $u,v\in J$ and $w\in \vset{D}\setminus J$ such that~$v$ and~$w$ both cite~$u$, and~$v$ does not cite~$w$, then delete~$w$ from~$D$ and decrease~$k$ by one.
\end{reductionrule}

To see that this rule is correct, observe that whenever such vertices $u$, $v$, $w$ exist, we have $w\in \citeB{v}{D-U}\cup \citeR{v}{D-U}$ for every $U\subseteq \vset{G}\setminus (J\cup \{w\})$. Thus, to achieve a {\cd} index of~$1$ for~$v$ by deleting papers from $\vset{D}\setminus J$, 
the paper~$w$ has to be deleted. 

\begin{reductionrule}
\label{reduction-rule-b}
    If there exist $u, v\in J$ and $w\in \vset{D}\setminus J$ such that~$v$ and~$u$ both cite~$w$, and~$u$ does not cite~$v$, then delete~$w$ from~$D$ and decrease~$k$ by one.
\end{reductionrule}

The correctness of the reduction follows a similar logic. 

After the reduction rules are exhaustively applied, we construct an undirected graph~$G$ 
as follows. 
The vertex set of~$G$ is $\vset{D}\setminus J$. For $w,w' \in V(G)$, we add an edge $\{w,w'\}$ to $G$ if and only if $w$ cites $w'$ and there exists $v \in J$ such that $v$ cites $w'$ but not $w$ in $D$. In this situation, at least one of $w$ or $w'$ must be deleted to ensure that the {\cd} index of $v$ becomes $1$. The edge represents such a conflict.

We now enumerate all vertex covers~$S$ of~$G$ of size at most~$k$. Using the folklore branching algorithm for {\prob{Vertex Cover}}, this can be done in time~$\bigos{2^k}$ (see, e.g.,~\cite{DBLP:conf/cocoon/Fernau02}). If there exists at least one such vertex cover~$S$ for which every paper in~$J$ has  {\cd} index~$1$ in $D-S$, we return  ``{\yes}''. Otherwise, we return ``\no''. 

We now verify the correctness. It is clear that if the algorithm returns ``\yes'', then~$I$ is a {\yesins}. 
Suppose the algorithm returns ``\no''. Then one of the following cases occurs.

\begin{description}
    \item[Case~1:] $G$ has no vertex cover of size at most~$k$.

    In this case, no matter which at most~$k$ papers from $\vset{D} \setminus J$ are deleted, there exist $w, w' \in \vset{G}$ and $v \in J$ such that $v$ cites $w'$ but not $w$, and $w$ cites $w'$ in~$D$. As a result,~$v$ cannot attain a {\cd} index of~$1$ by deleting at most~$k$ papers from $\vset{D} \setminus J$. Thus,~$I$ is a {\noins}.

    \item[Case~2:] Every vertex cover~$S$ of size at most~$k$ in~$G$ fails to make all papers in $J$ attain {\cd} index $1$ in $D - S$.

Let $S$ be a vertex cover of $G$ of size at most $k$, and let $v$ be a paper in $J$ whose {\cd} index in $D-S$ is less than~$1$. 
    Then there exist papers $w$ and $w'$ in $D-S$ such that $w'$ is cited by $v$ and $w$, and $v$ does not cite $w$. 
Four configurations are possible: 
    \begin{enumerate}
        \item[(2.1)] $w, w' \in J$,
        \item[(2.2)] $w, w' \in \vset{D} \setminus (S \cup J)$,
        \item[(2.3)] $w \in J$, $w' \in \vset{D} \setminus (S \cup J)$,
        \item[(2.4)] $w' \in J$, $w \in \vset{D} \setminus (S \cup J)$.
    \end{enumerate}
    Cases~(2.3) and~(2.4) are impossible due to the exhaustive application of the reduction rules. Case~(2.1) contradicts our initial assumption on~$D[J]$. Finally, Case~(2.2) contradicts the assumption that $S$ is a vertex cover of~$G$. 
    Thus, this case cannot occur. 
\end{description}

Therefore, if the algorithm returns ``\no'', then~$I$ is indeed a {\noins}.
%
\end{proof}

The algorithm in the proof of Theorem~\ref{thm-del-d-1-fpt-k} suggests a kernel of size at most $\abs{J}+k^2+k$ for {\deleteimprovement} with $d=1$. A kernelization for a parameterized problem is a polynomial-time algorithm that maps an instance $(I, \kappa)$ to an equivalent instance of the same problem whose size is bounded by a computable function of~$\kappa$. 

To obtain such a kernel, we augment the two reduction rules from the proof of Theorem~\ref{thm-del-d-1-fpt-k} with the following well-known reduction rule for {\prob{Vertex Cover}}.

\begin{reductionrule}
Let $G$ be the graph constructed after exhaustively applying Reduction Rules~\ref{reduction-rule-a} and~\ref{reduction-rule-b}, as in the proof of Theorem~\ref{thm-del-d-1-fpt-k}. If $G$ contains a vertex~$v$ of degree at least $k+1$, then delete $v$ from $D$ and decrease $k$ by one.
\end{reductionrule}

We conjecture that applying the state-of-the-art kernelization techniques for {\prob{Vertex Cover}} (see, e.g.,~\cite{DBLP:conf/birthday/FellowsJKRW18,DBLP:journals/tcs/LiZ18}) may yield a kernel of size $\abs{J} + 2k$ for {\deleteimprovement} with $d = 1$, potentially when restricted to acyclic citation digraphs. Obtaining a linear kernel parameterized solely by $k$ may be considerably more challenging. Since this direction falls outside the scope of the present paper, we leave it as an open question.

\begin{openquestion}
Does {\deleteimprovement} admit a polynomial kernel with respect to~$k$ when $d=1$ and the citation digraph is acyclic?
\end{openquestion}

For $d<1$, the complexity of {\deleteimprovement} changes significantly, as shown by the following theorem.

\begin{theorem}
\label{thm-delete-improvement-wah}
For any positive rational number $d<1$, {\deleteimprovement} is {\wah} parameterized by $\abs{J}+k$, even if the citation digraph is acyclic.
\end{theorem}

\begin{proof}
  We prove the theorem by a reduction from \prob{Multicolored Clique}. First, we present the reduction for a specific value of~$d$, and then extend it to any positive rational $d<1$.
    
    Let $I=(G, (V_1, V_2, \dots, V_{\kappa}))$ be an instance of {\prob{Multicolored Clique}}. Without loss of generality, assume that $\kappa\geq 2$. For each $i\in [\kappa]$, let $n_i=\abs{V_i}$. For each pair $\{i,j\}\subseteq [\kappa]$, let~$m_{\{i,j\}}$ denote the number of edges between~$V_i$ and~$V_j$ in~$G$. Let $n=\abs{\vset{G}}$ and let $m=\abs{\eset{G}}$.
    We create an instance of {\deleteimprovement} as follows.
\begin{description}
    \item[Citation digraph~$D$.]  For each vertex $v\in \vset{G}$, we create a paper denoted by the same symbol for simplicity. For each edge $e\in \eset{G}$, we create a paper~$p(e)$. For every pair $\{i,j\}\subseteq [\kappa]$, let
    \[P_{\{i,j\}}=\{p(e) \setmid e\in \eset{G}, e\cap V_i\neq\emptyset, e\cap V_j\neq\emptyset\}\] 
    be the set of papers corresponding to edges between~$V_i$ and~$V_j$ in~$G$. Thus, $\abs{P_{\{i,j\}}}=m_{\{i,j\}}$. Let $P(\eset{G})=\{p(e) \setmid e\in \eset{G}\}$ be the set of all edge-papers. Additionally, we create a set $X=\{x_1, x_2, \dots, x_{\kappa}\}$ of~$\kappa$ papers, a set $X'=\{x_1', x_2', \dots, x_{\kappa}'\}$ of~$\kappa$ papers, and a set $Y=\{y_{\{i,j\}}\setmid \{i,j\}\subseteq [\kappa]\}$ of ${\kappa \cdot (\kappa-1)}/{2}$ papers. Moreover, for every $i\in [\kappa]$, we create a set~$C_i$ of $m+n-n_i$ papers, and for every pair $\{i,j\}\subseteq [\kappa]$, we create a set~$C_{\{i,j\}}$ of $m+n-m_{\{i,j\}}$ papers. Finally, we create a paper~$z$. 
    Let $J=X\cup Y$.

    We create arcs in~$D$ so that exactly the following citations exist:
    \begin{itemize}
        \item For every $i\in [\kappa]$, all papers from $C_i\cup V_i \cup \{x_i\}$ cite the paper~$x_i'$.
        \item For every pair~$\{i,j\}\subseteq [\kappa]$, all papers from $C_{\{i,j\}}\cup \{y_{\{i,j\}}\}$ cite all papers from $V_i\cup V_j$.
        \item For every pair $\{i,j\}\subseteq [\kappa]$ and every edge $e=\edge{u}{v}$ between~$V_i$ and~$V_j$ in~$G$, the paper~$p(e)$ cites~$u$ and~$v$.
        \item The paper~$z$ cites all papers from $X\cup Y$.
    \end{itemize}
\end{description}
See Figure~\ref{fig-delete-improvement-wah} for an illustration of the construction. The acyclicity of the citation digraph~$D$ is witnessed by the topological ordering $(z, \overrightarrow{J}, \overrightarrow{C}, \overrightarrow{P(E(G))}, \overrightarrow{\vset{G}})$, where $C=(\bigcup_{i \in [\kappa]} C_i) \cup (\bigcup_{\{i,j\} \subseteq [\kappa]} C_{\{i,j\}})$. 
To complete the reduction, we set $d={1}/{(m+n)}$ and  $k=\kappa$.
The resulting instance of {\deleteimprovement} is $g(I)=(D, J, k, d)$ which  can be constructed in polynomial time. 
\begin{figure}
    \centering
    \includegraphics[width=0.4\textwidth]{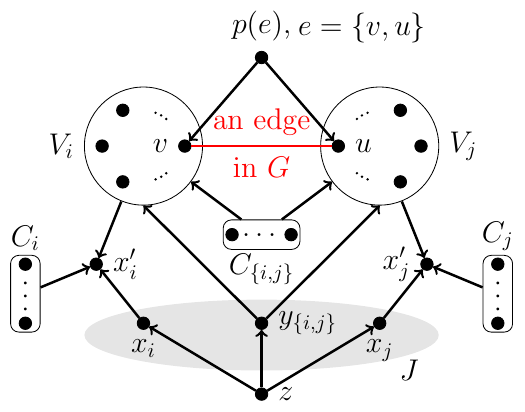}
    \caption{Illustration of the construction of the citation digraph in the proof of Theorem~\ref{thm-delete-improvement-wah}.}
    \label{fig-delete-improvement-wah}
\end{figure}
We now prove the correctness of the reduction.

    $(\Rightarrow)$ Suppose that $G$ contains a $\kappa$-colored clique $K$. After deleting all the papers in~$K$ from~$D$, the {\cd} index of every paper in~$X$ becomes 
    $\frac{\abs{\{z\}}}{\abs{\{z\}}+\abs{C_i}+\abs{V_i}-1}=d$, and that of every paper in~$Y$ becomes 
    $\frac{\abs{\{z\}}}{\abs{\{z\}}+\abs{C_{\{i,j\}}}+\abs{P_{\{i,j\}}}-1}=d$.

    $(\Leftarrow)$ Assume that there exists $V'\subseteq \vset{D}\setminus J$ with $\abs{V'}\leq k$ such that every paper in~$J$ has a {\cd} index of at least $d={1}/{(n+m)}$ in~$D-V'$. By construction,~$z$ is the only citation of papers in~$J$, and~$z$ does not cite any references of papers in~$J$. Hence $z\not\in V'$
    . 
    We show below that $V'$ does not contain any edge-paper or any paper in $\bigcup_{\{i,j\}\subseteq [\kappa]} C_{\{i,j\}}$. 
    \begin{claim}
        \label{claim-a}
        $V'\cap \left(\left(\bigcup_{\{i,j\}\subseteq [\kappa]} C_{\{i,j\}}\right)\cup P(\eset{G})\right)= \emptyset$.
    \end{claim}

    \begin{proof}
    Suppose otherwise. Since $\abs{V'}\leq k=\kappa$, there exists $i\in [\kappa]$ such that $V'\cap (V_i\cup C_i\cup \{x_i'\})=\emptyset$. By  construction, the {\cd} index of $x_i$ in $D-V'$ is
    $\frac{\abs{\{z\}}}{\abs{\{z\}}+\abs{V_i\cup C_i}}=\frac{1}{n+m+1}<d$, a contradiction.
\end{proof}

Next, we prove that $V'$ intersects every $V_i$.

    \begin{claim}
    \label{claim-b}
        $V'\cap V_i\neq \emptyset$ for all $i\in [\kappa]$.
    \end{claim}

    \begin{proof}
        Assume for contradiction that there exists $i\in [\kappa]$ such that $V'\cap V_i=\emptyset$. Since $\kappa\ge 2$, there exists $j\in[\kappa]\setminus\{i\}$. By Claim~\ref{claim-a}, the {\cd} index of~$y_{\{i,j\}}$ in $D-V'$ is
    \begin{equation}
    \label{eq-l}
        \frac{\abs{\{z\}}}{\abs{\{z\}}+\abs{C_{\{i,j\}}}+\abs{P_{\{i,j\}}}}=\frac{1}{m+n+1}<d,
    \end{equation}
    a contradiction.
    \end{proof}

    By Claim~\ref{claim-b} and the fact that $\abs{V'}\leq k=\kappa$, it follows that $V'$ contains exactly one vertex from each $V_i$ for all $i\in [\kappa]$. We show that $V'$ forms a clique in~$G$. Suppose, for contradiction, that there exist $v\in V'\cap V_i$ and $u\in V'\cap V_j$ for some $\{i,j\}\subseteq [\kappa]$ that are nonadjacent in~$G$. Consequently, every paper in $P_{\{i,j\}}$ cites at least one paper in $(V_i\cup V_j)\setminus \{v,u\}$. Since all papers in $V_i\cup V_j$ are references of~$y_{\{i,j\}}$, by the construction of~$D$ and Claim~\ref{claim-a}, the {\cd} index of $y_{\{i,j\}}$ in $D-V'$ is less than $d$ (see~Inequality~\eqref{eq-l}), a contradiction. 
    Thus, $V'$ is a clique in~$G$, and by Claim~\ref{claim-b}, $I$ is a {\yesins}. 

    To establish the hardness for any positive rational number $d<1$ with the canonical form~${p}/{q}$, we substitute~$z$ with a set~$Z$ of $\frac{p}{q-p}\cdot (m+n-1)$ papers, each citing all papers in~$J$. We assume that $m+n-1$ is divisible by $q-p$, which does not change the {\wahns} of {\prob{Multicolored Clique}}.\footnote{If this condition is not satisfied, we introduce isolated vertices (assigning them arbitrarily to $V_1$, $\dots$, $V_{\kappa}$) until it holds. The resulting instance is equivalent to the original one.} The correctness argument parallels the proof described above.
\end{proof}

Now we examine the parameter~$\abs{J}$. 
In particular, we present a complexity classification given by the following three theorems.

\begin{theorem}
\label{thm-del-d-1-p-j-1}
    For $d=1$, {\deleteimprovement} is polynomial-time solvable if $\abs{J}=1$.
\end{theorem}

\begin{proof}
    Let $I=(D, J, k,d)$ be an instance of {\deleteimprovement}, where $d=1$ and $J=\{v^{\star}\}$ with $v^{\star}\in \vset{D}$. We compute the four sets $\outneighbor{v^{\star}}{D}$, $\citeF{v^{\star}}{D}$,~$\citeB{v^{\star}}{D}$, and~$\citeR{v^{\star}}{D}$
    , and distinguish between two cases.

    \begin{description}
       \item[Case~1:] $\citeF{v^{\star}}{D}\neq \emptyset$. \hfill

Let $V'=\citeB{v^{\star}}{D}\cup \citeR{v^{\star}}{D}$.
We construct a bipartite graph~$G$ with bipartition $(\outneighbor{v^{\star}}{D}, V')$, where an edge between $v\in V'$ and  $u\in \outneighbor{v^{\star}}{D}$ exists if and only if $v$ cites $u$ in~$D$. To ensure that~$v^{\star}$ attains a {\cd} index of~$1$ by deleting papers, for each $v\in V'$, either~$v$ must be deleted or all its neighbors in~$G$ must be deleted. This is equivalent to selecting a vertex cover of~$G$. Therefore, we return ``\yes'' if and only if $G$ admits a vertex cover of size at most $k$. 
Since a minimum vertex cover in bipartite graphs can be computed in polynomial time~\cite{Koenig1916,LPmatchingtheory1986}, this case can be solved in polynomial time.

       \item[Case~2:] $\citeF{v^{\star}}{D}=\emptyset$. \hfill

       If $\citeB{v^{\star}}{D}=\emptyset$, then $v^{\star}$ has no citations in~$D$, and we return ``\no''. 
       Otherwise, we delete papers from~$\outneighbor{v^{\star}}{D}$ so that there exists a $v\in \citeB{v^{\star}}{D}$ that cites $v^{\star}$ but none of its references. 
       We enumerate all such choices of~$v$. For each enumerated $v\in \citeB{v^{\star}}{D}$, we:
       \begin{enumerate}[label={(\arabic*)}]
           \item  decrease~$k$ by $\abs{\outneighbor{v}{D}\cap \outneighbor{v^{\star}}{D}}$, and delete all papers in $\outneighbor{v}{D}\cap \outneighbor{v^{\star}}{D}$ from~$D$. 
           \item apply the procedure from Case~1.
       \end{enumerate}
       If any branch returns ``\yes'', then $I$ is a {\yesins}; otherwise, $I$ is a {\noins}.
    \end{description}
The correctness of the algorithm and its polynomial running time are straightforward to verify. 
\end{proof}


The next result shows that this tractability is fragile: even with $d=1$, increasing the number of targeted papers from one to two already leads to computational intractability.

\begin{theorem}
\label{thm-del-nph-j-2-d-1}
    For $d=1$, {\deleteimprovement} is {\nph}, even if $\abs{J}=2$ and the citation digraph is acyclic and tripartite.
\end{theorem}

\begin{proof}
    We prove the theorem via a reduction from {\prob{Vertex Cover}} on tripartite graphs
    . Let $I=(G, \kappa)$ be an instance of {\prob{Vertex Cover}}, where~$G$ is tripartite with the tripartition $(V_1,V_2,V_3)$. We construct an instance $g(I)=(D, J, k, d)$ of {\deleteimprovement} with $d=1$ as follows. See Figure~\ref{fig-del-nph-j-2-d-1} for an illustration.
    \begin{description}
        \item[Citation digraph~$D$.] Let $\vset{D}=V_1\cup V_2\cup V_3\cup \{v_1^{\star},v_2^{\star}, x\}$. The arcs of~$D$ are defined as follows:
        \begin{itemize}
            \item The paper $x$ cites both $v_1^{\star}$ and $v_2^{\star}$.
            \item The paper $v_1^{\star}$ cites all papers from $V_1$.
            \item The paper $v_2^{\star}$ cites all papers from $V_2$.
            \item For every edge $\edge{v}{u}\in \eset{G}$ with $v\in V_3$ and $u\in V_1\cup V_2$,~$v$ cites~$u$.
            \item For every edge $\edge{v}{u}\in \eset{G}$ with $v\in V_1$ and $u\in V_2$,~$v$ cites~$u$.
        \end{itemize}
        \end{description}
        The citation digraph~$D$ is acyclic with the topological ordering $(x, v_1^{\star}, v_2^{\star}, \overrightarrow{V_3}, \overrightarrow{V_1}, \overrightarrow{V_2})$. 
Moreover,~$D$ admits a tripartition $\left((V_3 \cup \{v_1^{\star}, v_2^{\star}\}), (V_1 \cup \{x\}), V_2\right)$. 
Finally, we set $J=\{v_1^{\star}, v_2^{\star}\}$ and $k=\kappa$.  We now prove the correctness of the reduction. Note that $\citeB{v_1^{\star}}{D}=\citeB{v_2^{\star}}{D}=\emptyset$.

        \begin{figure}
            \centering
            \includegraphics[width=0.28\textwidth]{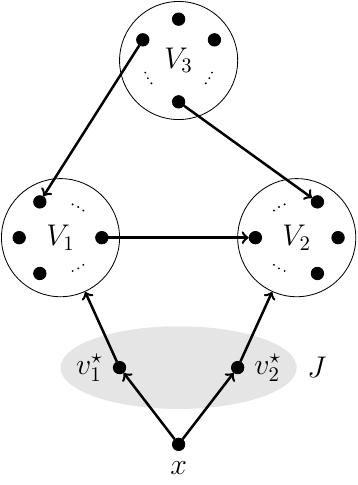}
            \caption{Illustration of the reduction used in the proof of Theorem~\ref{thm-del-nph-j-2-d-1}. 
            }
            \label{fig-del-nph-j-2-d-1}
        \end{figure}

        $(\Rightarrow)$ Assume that~$G$ has a vertex cover~$S$ of at most~$\kappa$ vertices. Then, for both $i \in [2]$, we have $\citeF{v_i^{\star}}{D-S} = \{x\}$ and $\citeR{v_i^{\star}}{D-S} = \emptyset$. Thus, both~$v_1^{\star}$ and~$v_2^{\star}$ have {\cd} index~$1$ in $D-S$, and $g(I)$ is a {\yesins}.

        $(\Leftarrow)$ Assume that there exists $S\subseteq \vset{D}\setminus J$ with~$\abs{S}\leq k$ such that both~$v_1^{\star}$ and~$v_2^{\star}$ have {\cd} index~$1$ in $D-S$. Clearly, $x \notin S$, since~$x$ is the only citation of~$v_1^{\star}$ and~$v_2^{\star}$
        . Thus, $S\subseteq \vset{G}$. Since $k=\kappa$, it suffices to show that~$S$ is a vertex cover of~$G$. Assume for contradiction that $G-S$ contains an edge~$\edge{v}{u}$. If this edge is between~$V_3$ and~$V_i$ where $i\in [2]$, by construction, the endpoint of the edge in~$V_3$ lies in $\citeR{v_i^{\star}}{D-S}$, implying that the {\cd} index of $v_i^{\star}$ is not~$1$ in $D-S$, a contradiction. If the edge is between~$V_1$ and~$V_2$, say $v\in V_1$ and $u\in V_2$, then by construction,~$v$ cites~$u$, and~$u$ is a reference of~$v_2^{\star}$. Hence, $v \in \citeR{v_2^{\star}}{D-S}$, again contradicting the assumption that the {\cd} index of~$v_2^{\star}$ in $D-S$ is~$1$.
\end{proof}

While Theorems~\ref{thm-del-d-1-p-j-1} and \ref{thm-del-nph-j-2-d-1} concern the case $d=1$, the situation changes dramatically for $d<1$.The following theorem shows that, in this regime, the problem becomes {\wah} with respect to~$k$, even when $\abs{J}=1$ and the citation digraph is acyclic and bipartite. 

\begin{theorem}
\label{thm-delete-improvement-j-paranp-hard}
For any positive rational number $d<1$,  {\deleteimprovement} is {\wah} with respect to~$k$, even if $\abs{J}=1$ and the citation digraph is acyclic and bipartite.
\end{theorem}

\begin{proof}
    We prove the theorem via a reduction from {\prob{Clique}}. Let $d<1$ be a fixed positive rational number, and write $d=p/q$ in canonical form. Let $I=(G,\kappa)$ be an instance of {\prob{Clique}} with $m=\abs{\eset{G}}$. Without loss of generality, assume $m>{\kappa\cdot (\kappa-1)}/{2}$. We construct an instance of {\deleteimprovement} as follows.
\begin{description}
    \item[Citation digraph~$D$.] For every vertex $v\in \vset{G}$, we create a paper denoted by the same symbol. For each edge $e\in \eset{G}$, we create a set~$P(e)$ of $q-p$ papers. Let $P(\eset{G})=\bigcup_{e\in \eset{G}}P(e)$
    . We have $\abs{P(\eset{G})}=(q-p)\cdot m$. In addition, we create a paper $v^{\star}$, and a set $X$ of $p \cdot \left(m-\frac{\kappa\cdot (\kappa-1)}{2}\right)$ papers. 
    The arcs of~$D$ are defined as follows:
    \begin{itemize}
        \item All papers in $X$ cite the paper $v^{\star}$.
        \item The paper $v^{\star}$ cites all papers in~$\vset{G}$.
        \item For every edge $e=\edge{v}{u}$ in~$G$, all papers in~$P(e)$ cite both~$v$ and~$u$.
    \end{itemize}
\end{description}
It is easy to see that~$D$ is acyclic and bipartite. Setting $J=\{v^{\star}\}$ and $k=\kappa$, the resulting instance of {\deleteimprovement} is $g(I)=(D, J, k, d)$. We now prove the correctness of the reduction. 

$(\Rightarrow)$ Assume that~$G$ contains a $\kappa$-clique~$K$. Let $D'=D-K$. Let $P'=\bigcup_{e\in \eset{G[K]}}P(e)$ be the set of papers created for edges whose endpoints both lie in~$K$. Then $\abs{P'}=(q-p)\cdot {\kappa\cdot (\kappa-1)}/{2}$. 
The {\cd} index of~$v^{\star}$ in~$D'$ is
$\frac{\abs{X}}{\abs{X}+\abs{P(\eset{G})}-\abs{P'}}$. Substituting the values of~$X$, $\abs{P'}$, and $\abs{P(\eset{G})}$ yields ${p}/{q}$.
Hence,~$g(I)$ is a {\yesins}.

$(\Leftarrow)$ Assume that there is a subset $U\subseteq \vset{D}\setminus J$ of at most~$k$ papers such that the {\cd} index of~$v^{\star}$ in $D'=D-U$ is at least~${p}/{q}$. Let $K=U\cap \vset{G}$ and let $U'=U\cap P(\eset{G})$. Clearly,
\begin{equation}
    \label{eq-uu}
    \abs{K}+\abs{U'}\leq \abs{U}\leq k=\kappa.
\end{equation}
Let~$E'$ denote the set of edges in~$G$ whose endpoints both lie in~$K$. Then
\begin{equation}
    \label{eq-ee}
    \abs{E'}\leq\frac{\abs{K}\cdot (\abs{K}-1)}{2}.
\end{equation}
Let $P(E')=\bigcup_{e\in E'}P(e)$ denote the set of papers corresponding to edges in $E'$.
By the construction of $D$, 
\begin{equation}
    \label{eq-pe}
    \abs{P(E')}=(q-p)\cdot \abs{E'}.
\end{equation}
Moreover, the {\cd} index of $v^{\star}$ in $D'$ is
\begin{equation}
\begin{split}
\dindex{v^{\star}}{D'}= \frac{\abs{X\setminus U}}{\abs{X\setminus U}+\abs{P(\eset{G})\setminus (U'\cup P(E'))}} & 
\leq   \frac{\abs{X\setminus U}}{\abs{X\setminus U}+\abs{P(\eset{G})}- \abs{U'}-\abs{P(E')}}\\
&             \leq \frac{\abs{X}}{\abs{X}+m\cdot (q-p)- \abs{U'}-\abs{P(E')}}.
\end{split}
\end{equation}
By assumption, 
\begin{equation}
    \label{eq-cdindex}
    \dindex{v^{\star}}{D'}\geq {p}/{q}.
\end{equation}
A straightforward calculation shows that the inequalities in~\eqref{eq-uu}--\eqref{eq-cdindex} can hold only if 
$\abs{K}=\kappa$, $\abs{E'}={\kappa \cdot (\kappa-1)}/{2}$, and $K=U\cap \vset{D}$ (i.e., $U'=\emptyset$ and $X\setminus U=X$). 
Thus $K$ is a $\kappa$-clique in~$G$, and~$I$ is a {\yesins}.
\end{proof}

We now consider the variant {\deleteimprovementinf}, where the number of paper deletions is unrestricted. Our first result shows that the problem remains computationally hard even under seemingly restrictive conditions.

\begin{theorem}
\label{del-inf-np-hard-j-3}
    {\deleteimprovementinf} is {\nph} even if $\abs{J}=3$ and every paper not in $J$ is either a reference or a citation of some paper in $J$.
\end{theorem}

\begin{proof}
    We prove the theorem by modifying the reduction in the proof of Theorem~\ref{thm-del-nph-j-2-d-1}:
    \begin{itemize}
        \item We add a paper $v^{\star}_3$ to $J$, so that $J=\{v^{\star}_1, v^{\star}_2, v^{\star}_3\}$.
        \item We let $v^{\star}_3$ cite $v^{\star}_1$, and let it be cited by all papers in~$\vset{G}$.
        \item We let $v_2^{\star}$ cite $v_1^{\star}$.
        \item Let~$n=\abs{\vset{G}}$. 
        We replace~$x$ with a set $X$ of $n-\kappa$ papers, each citing only $v^{\star}_2$.
        \item We set $d=\frac{n-\kappa}{n-\kappa+1}$.
    \end{itemize}
    Without loss of generality, assume that $n-\kappa\geq 3$; then $d>2/3$. See Figure~\ref{fig-del-inf-np-hard-j-3} for an illustration of the reduction.
    \begin{figure}
        \centering
        \includegraphics[width=0.3\textwidth]{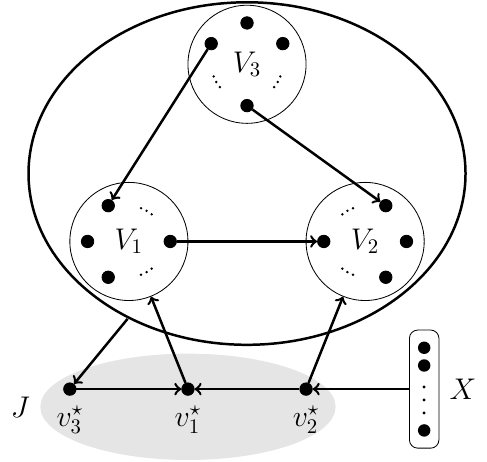}
        \caption{Illustration of the reduction in the proof of Theorem~\ref{del-inf-np-hard-j-3}.}
        \label{fig-del-inf-np-hard-j-3}
    \end{figure}
    Let~$g(I)$ denote the constructed instance of {\deleteimprovementinf}. 
     We prove its correctness below.

     $(\Rightarrow)$ Assume that~$G$ has a vertex cover $S$ of at most $\kappa$ vertices. Consider the citation digraph $D'=D-S$. Since~$S$ is a vertex cover, no paper cites any reference of~$v^{\star}_1$ in~$D'$, i.e., $\citeB{v^{\star}_1}{D'}\cup \citeR{v^{\star}_1}{D'}=\emptyset$. As~$v^{\star}_2$ and~$v^{\star}_3$ both cite~$v^{\star}_1$, we know that the {\cd} index of $v^{\star}_1$ in~$D'$ is~$1$. Furthermore, no paper in $\vset{G}-S$ cites any paper in $V_2\setminus S$, the set of references of~$v^{\star}_2$ in~$D'$ that lie in~$\vset{G}$. Hence, the {\cd} index of~$v^{\star}_2$ in~$D'$ is
     $\frac{\abs{X}}{\abs{X}+\abs{\{v^{\star}_3\}}}=\frac{n-\kappa}{n-\kappa+1}=d$. Finally, the {\cd} index of~$v^{\star}_3$ in~$D'$ is
     \[\frac{\abs{\vset{G}\setminus S}}{\abs{\vset{G}\setminus S}+\abs{\{v^{\star}_2\}}}\geq \frac{n-\kappa}{n-\kappa+1}=d.\] 
     Therefore,~$g(I)$ is a {\yesins}.

     $(\Leftarrow)$ Assume that there exists $V'\subseteq \vset{D}\setminus J$ such that the {\cd} index of each paper in~$J$ in the citation digraph $D'=D-V'$ is at least $d$. Let $S=V'\cap \vset{G}$.
     
     \begin{claim}
     \label{claim-aaa}
         $S$ is a vertex cover of $G$.
     \end{claim}
     
     \begin{proof}
         Assume for contradiction that $G-S$ contains an edge~$\edge{v}{u}$. Two cases are possible.
         \begin{description}
             \item[Case~1:] $\edge{v}{u}$ is between $V_3$ and $V_1$. \hfill

             Let $v\in V_3$ and $u\in V_1$. By construction,~$u$ is a reference of~$v^{\star}_1$,~$v$ cites~$u$, and~$v$ and~$v^{\star}_1$ do not cite each other. That is, $v\in \citeR{v^{\star}_1}{D'}$. Then, the {\cd} index of $v^{\star}_1$ in $D'$ is
             $\frac{\abs{\{v^{\star}_3,v^{\star}_2\}}}{\abs{\{v^{\star}_3,v^{\star}_2\}}+\abs{\citeR{v^{\star}_1}{D'}}}\leq 
             \frac{2}{3}<d$, a contradiction.

             \item[Case~2:] $\edge{v}{u}$ is between $V_1\cup V_3$ and $V_2$. \hfill

             Let $v\in V_3\cup V_1$ and $u\in V_2$. By construction,~$u$ is a reference of~$v^{\star}_2$,~$v$ cites~$u$ but does not cite~$v^{\star}_2$. Thus the {\cd} index of~$v^{\star}_2$ in~$D'$ is at most
             $\frac{\abs{X\setminus V'}}{\abs{X\setminus V'}+\abs{\{v^{\star}_3, v\}}}\leq
              \frac{\abs{X}}{\abs{X}+\abs{\{v^{\star}_3, v\}}}=
              \frac{n-\kappa}{n-\kappa+2}<d$, a contradiction.
         \end{description}
         Thus~$S$ is a vertex cover of~$G$.
     \end{proof}
The {\cd} index of $v^{\star}_3$ in~$D'$ is
    $\frac{\abs{\vset{G}\setminus S}}{\abs{\vset{G}\setminus S}+\abs{\{v^{\star}_2\}}}$. 
     Since this value is at least $d=\frac{n-\kappa}{n-\kappa+1}$, it follows that $\abs{S}\leq \kappa$. 
   By Claim~\ref{claim-aaa}, $I$ is a {\yesins}.
     \end{proof}

The above hardness result relies on the presence of directed cycles in the citation digraph. The complexity of {\deleteimprovementinf} on acyclic citation digraphs remains open. 

\begin{openquestion}
    Is {\deleteimprovementinf} {\nph} when the citation digraph is acyclic?
\end{openquestion}

Despite this, we obtain positive results when the number of distinguished papers is small and the citation digraph is acyclic. 
In particular, when there is only one distinguished paper~$v^{\star}$, the problem can be solved trivially. If~$v^{\star}$ has at least one citation, we can assign it the maximal {\cd} index of $1$ by deleting all papers except~$v^{\star}$ and an arbitrary paper citing it. 
We extend the polynomial-time solvability to the case where $\abs{J} = 2$. The following lemma is useful. 

\begin{lemma}
    \label{lem-remove-unrelated}
Let $D$ be a citation digraph and let $u,v \in \vset{D}$ such that $u$ does not cite $v$, and $v$ has at least one non-reciprocal citation 
(i.e., $\citeF{v}{D}\cup \citeB{v}{D}\neq \emptyset$). 
Then, 
$\abs{\dindex{v}{D-\{u\}}}\geq \abs{\dindex{v}{D}}$.
\end{lemma}

\begin{proof}
    Let $D'=D-\{u\}$. By definition, $\dindex{v}{D}=\frac{\abs{\citeF{v}{D}}-\abs{\citeB{v}{D}}}{\abs{\citeF{v}{D}}+\abs{\citeB{v}{D}}+\abs{\citeR{v}{D}}}$. Since~$u$ does not cite~$v$, it holds $u\not \in ({\citeB{v}{D}\cup \citeF{v}{D}})$. Consequently,
    \[\dindex{v}{D'}=\frac{\abs{\citeF{v}{D}}-\abs{\citeB{v}{D}}}{\abs{\citeF{v}{D}}+\abs{\citeB{v}{D}}+\abs{\citeR{v}{D}\setminus \{u\}}}.\] Since $v$ has at least one non-reciprocal citation in~$D$, both $\dindex{v}{D}$ and $\dindex{v}{D'}$ are well-defined. 
    If $\abs{\citeF{v}{D}}\geq \abs{\citeB{v}{D}}$, $\dindex{v}{D'}\geq \dindex{v}{D}$; otherwise, $\dindex{v}{D'}\leq \dindex{v}{D}$, with equality when $u\not \in \citeR{v}{D}$.
\end{proof}

Using the above lemma, we obtain the following result.

\begin{theorem}
\label{thm-inf-poly-j-2}
For any positive rational number $d$,  {\deleteimprovementinf} is polynomial-time solvable when $\abs{J}\leq 2$ and the citation digraph is acyclic.
\end{theorem}

\begin{proof}
    Let $I=(D, J, d)$ be an instance of {\deleteimprovementinf} with  $\abs{J}\leq 2$ and $d$  a positive rational number.

    We first consider the case $\abs{J}=1$, where $J=\{v^{\star}\}$. The algorithm proceeds as follows: delete all references of~$v^{\star}$ in $\vset{D}\setminus J$, and return ``\yes'' if the {\cd} index of~$v^{\star}$ in the resulting citation digraph is at least~$d$, and ``\no'' otherwise.

    We now consider the case $\abs{J}=2$, where~$J=\{v^{\star}_1, v^{\star}_2\}$. If $\citeF{v^{\star}_1}{D}\cup \citeB{v^{\star}_1}{D}= \emptyset$ or $\citeF{v^{\star}_2}{D} \cup \citeB{v^{\star}_2}{D}= \emptyset$, we return ``{\no}''. Otherwise, by Lemma~\ref{lem-remove-unrelated}, we remove from~$D$ all papers in $\vset{D}\setminus J$ citing neither~$v^{\star}_1$ nor~$v^{\star}_2$. Hence, we may assume the following.
    \medskip
    
\noindent{\bf{Assumption}.} {\it{Every paper in $\vset{D}\setminus J$ cites at least one paper in~$J$, and either paper in~$J$ admits at least one citation.}}
\medskip

We consider first the case where no arc exists between~$v^{\star}_1$ and~$v^{\star}_2$ by partitioning $\vset{D}\setminus J$ into three sets:
\begin{itemize}
    \item $B$ is the set of papers citing both $v^{\star}_1$ and $v^{\star}_2$,
    \item $A$ is the set of papers citing~$v^{\star}_1$ but not~$v^{\star}_2$, and
    \item $C$ is the set of papers citing~$v^{\star}_2$ but not~$v^{\star}_1$.
\end{itemize}
If $B\neq\emptyset$, return ``\yes'', since after deleting all papers except those in $J\cup B$, both papers in~$J$ have a {\cd} index of~$1$. 
Assume now $B=\emptyset$. Then, $A\neq \emptyset$ and $C\neq\emptyset$. We further divide~$A$ into two sets~$A_1$ and~$A_2$, where~$A_2$ consists of papers in~$A$ cited by~$v^{\star}_2$, and~$A_1=A\setminus A_2$. Analogously, let $C_2$ be the set of papers in $C$ cited by $v^{\star}_1$, and let $C_1=C\setminus C_2$. Since~$D$ is acyclic, at least one of~$A_2$ and~$C_2$ is empty. By symmetry, assume $C_2=\emptyset$. It follows that $C_1\neq\emptyset$. If $A_1\neq\emptyset$, we return ``\yes''. In this case, after deleting all papers except those in $J\cup A_1\cup C_1$, the {\cd} index of both papers in~$J$ is~$1$. Otherwise, $A_2\neq \emptyset$. If there exist $v\in C_1$ and $u\in A_2$ with~$v$ not citing~$u$, we return ``\yes''. In this case, after deleting all papers except those in $J\cup \{v,u\}$, both papers from~$J$ have a {\cd} index of~$1$. Otherwise, every paper in~$C_1$ cites every paper in~$A_2$. Then  $C_1=\citeB{v^{\star}_2}{D}$ and $\citeF{v^{\star}_2}{D}=\emptyset$, and thus $v^{\star}_2$ cannot attain a positive {\cd} index by deleting papers from $\vset{D}\setminus J$. Therefore, we return ``\no''. 

We now consider the case where there is an arc between $v^{\star}_1$ and~$v^{\star}_2$. By symmetry, assume that $v^{\star}_1$ cites~$v^{\star}_2$. Let $A_1$, $A_2$, $B$, $C_1$, and~$C_2$ be as defined previously
. Since~$D$ is acyclic, $A_2=\emptyset$. An illustration is given in Figure~\ref{fig-citation-j-2}.
\begin{figure}
    \centering
    \includegraphics[width=0.2\textwidth]{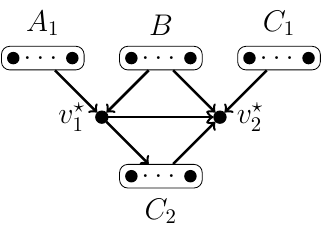}
     \caption{Illustration of the citation digraph when~$v^{\star}_1$ cites~$v^{\star}_2$ in the proof of Theorem~\ref{thm-inf-poly-j-2}.}
     \label{fig-citation-j-2}
\end{figure}
The algorithm proceeds as follows. 
\begin{itemize}
    \item If $A_1\neq \emptyset$, return ``\yes'', since in $D[J\cup A_1]$ both $v^{\star}_1$ and $v^{\star}_2$ have a {\cd} index of $1$. 
    \item Otherwise, return ``\no''. In this case, only papers in~$B$ cite $v^{\star}_1$, and each of them also cites $v^{\star}_2$, a reference of~$v^{\star}_1$. Hence, $\citeF{v^{\star}_1}{D'}=\emptyset$ for all subgraphs~$D'$ of~$D$, and for every $u\in B$ it holds that $u\in \citeB{v^{\star}_1}{D'}$ for all subgraphs~$D'$ of~$D$ containing~$J\cup \{u\}$. Thus, the {\cd} index of~$v^{\star}_1$ cannot be positive in any subgraph of~$D$ containing~$J$. 
\end{itemize}
The running time of the above algorithms is polynomial, which completes the proof.
\end{proof}




Finally, we consider the special threshold case $d = 1$. Here, the problem exhibits more nuanced complexity behavior. On the one hand, we prove a {\wahns} lower bound with respect to the parameter $\abs{J}$, even for acyclic citation digraphs (Theorem~\ref{thm-d-1-delete-inf-wah-ell}). On the other hand, we provide an {\xp} algorithm with respect to $\abs{J}$ (Theorem~\ref{thm-del-d-1-xp-j}), demonstrating that the problem is solvable in polynomial time for fixed $\abs{J}$.

\begin{theorem}
\label{thm-d-1-delete-inf-wah-ell}
For $d=1$, {\deleteimprovementinf} is {\wah} with respect to~$\abs{J}$, even when the citation digraph is acyclic, and every paper outside~$J$ either cites or is cited by some paper in~$J$.
\end{theorem}

\begin{proof}
 We prove the theorem via a reduction from {\prob{Multicolored Clique}}. Given an instance $I=(G, (V_1, V_2, \dots, V_{\kappa}))$ of {\prob{Multicolored Clique}}, we construct an instance $g(I)=(D, J, d)$ of {\deleteimprovementinf} with $d=1$. An illustration is provided in Figure~\ref{fig-d-cd-del-inf-wah-j}. 

\begin{description}
    \item[Citation digraph~$D$.] For each vertex $v\in \vset{G}$, we introduce a paper denoted by the same symbol. In addition,  we create a paper~$z$, a set $X=\{x_1,\dots, x_{\kappa}\}$ of $\kappa$ papers, and a set
$Y=\{y_{\{i,j\}}\}_{\{i,j\}\subseteq [\kappa]}$ of ${\kappa\cdot (\kappa-1)}/{2}$ papers. 
The arcs of~$D$ are defined by the following citations:
     \begin{itemize}
         \item Paper $z$ cites all papers in $Y$.
         \item For all $i\in [\kappa]$, all papers in $V_i$ cite the paper $x_i$.
         \item For all pairs $\{i,j\}\subseteq [\kappa]$, the paper $y_{\{i,j\}}$ cites all papers in $V_i\cup V_j$.
         \item For all $i,j\in[\kappa]$ with $i<j$, and all $u\in V_i$, $v\in V_j$ that are nonadjacent in~$G$, the paper~$u$ cites~$v$.
     \end{itemize}
The citation digraph~$D$ is acyclic with topological ordering $(z, \overrightarrow{Y}, \overrightarrow{V_1}, \overrightarrow{V_2},\dots, \overrightarrow{V_{\kappa}}, \overrightarrow{X})$. 
\end{description}
Let $J=X\cup Y$. Clearly, $\abs{J}=\kappa+\frac{\kappa \cdot (\kappa-1)}{2}$. 
\begin{figure}
    \centering
    \includegraphics[width=0.4\textwidth]{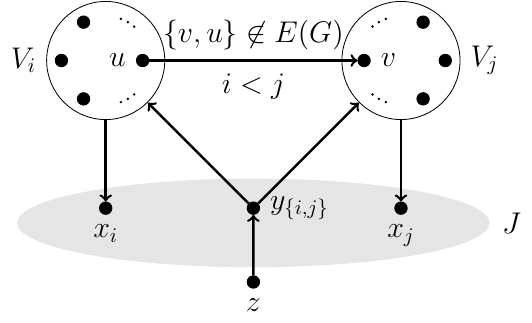}
    \caption{Illustration of the reduction in the proof of Theorem~\ref{thm-d-1-delete-inf-wah-ell}.}
    \label{fig-d-cd-del-inf-wah-j}
\end{figure}

$(\Rightarrow)$ If~$G$ has a $\kappa$-colored clique~$K$, then deleting all papers in $\vset{G}\setminus K$ from $D$ yields that every paper in $J$ has {\cd} index~$1$.

$(\Leftarrow)$ Assume that there exists $V'\subseteq \vset{D}\setminus J$ such that every paper in~$J$ has {\cd} index~$1$ in $D-V'$. Note that $z\not\in V'$, since~$z$ is the only citation of papers in~$Y$
. Moreover, each paper~$x_i$, where $i\in [\kappa]$, may have a {\cd} index of~$1$ in $D-V'$ only if $V_i\setminus V'\neq\emptyset$. We arbitrarily select one paper from each $V_i\setminus V'$, and let~$K$ be the set of the $\kappa$ selected papers. We claim that~$K$ is a clique in~$G$. Assume for contradiction that two papers $u, v \in K$ are nonadjacent in $G$. Let $u \in V_i$ and $v \in V_j$ with $i < j$. By construction, $u$ cites $v$. Since $v$ is a reference of $y_{\{i,j\}}$, the citation from $u$ to $v$ means that~$y_{\{i,j\}}$ cannot have {\cd} index $1$ in $D - V'$, a contradiction. 
\end{proof}


\begin{theorem}
\label{thm-del-d-1-xp-j}
  For $d=1$, {\deleteimprovementinf}  is in {\xp} with respect to~$\abs{J}$. 
\end{theorem}

\begin{proof}
    Let $I=(D, J, d)$ be an instance of {\deleteimprovementinf} with $d=1$. We solve $I$ using the polynomial-time algorithm for {\addimprovement} when $\abs{J}$ is constant (Theorem~\ref{thm-d-1-add-xp-j}). 
Observe that deleting papers from $\vset{D}\setminus J$ is equivalent to 
selecting a subset $U \subseteq \vset{D}\setminus J$ of papers to retain. 
Hence, $I$ is equivalent to determining whether there exists a subset 
$U \subseteq \vset{D}\setminus J$ such that every paper in~$J$ has 
{\cd} index~$1$ in $D[J \cup U]$. Therefore, we return ``\yes'' if and 
only if there exists $k \in \{0, \dots, \abs{V(D) \setminus J}\}$ such that $(D, (J,\vset{D}\setminus J), J, k, d)$ is a {\yesins} of {\addimprovement}. 
\end{proof}

\section{Concluding Remarks}
\label{sec-conclusion}
In an era of paper inflation, the use of qualitative indices to valuate the performance of researchers is becoming increasingly popular. As a relatively new metric, the {\cd} index has gained widespread attention in related studies over the past few years. Prior research has primarily focused on experimental analyses—identifying influential papers with higher {\cd} indices in specific research subfields, tracking trends in {\cd} indices, and statistically examining factors (e.g., the number of authors, agenda balance, geometric distribution of authors, etc.) that influence the trends in {\cd} indices or a paper’s {\cd} index, among others.  
These studies have provided valuable insights into the role of the {\cd} index in assessing academic impact. However, its reliability has also been questioned due to the lack of rigorous theoretical analysis. Without a solid theoretical foundation, concerns remain about its susceptibility to manipulation and its overall validity as a bibliometric measure. 

In this work, we initiated the study of the {\cd} index manipulation problems, offering a comprehensive understanding of their parameterized complexity for several significant parameters. See Tables~\ref{tab-summary-merge} and~\ref{tab-summary-add-deletion} for a summary of our concrete results. In general, we showed that these problems are hard to solve, with only a few exceptions where certain parameters are small numbers. This conveys that, at least from a theoretical point of view, the {\cd} index exhibits a high degree of reliability in the aspect of preventing manipulation. 

Rather than concluding at this juncture, we aspire for our initiative to ignite further comprehensive investigations aimed at fully addressing questions surrounding the applicability of the CD index. Specifically, we suggest the following topics for future research:

\begin{itemize}
    \item A traditional methodology for evaluating a measurement is the study of its axiomatic properties. Therefore, the first direction for future research is to explore which desired properties related to the measurement of the novelty of a paper are satisfied by the {\cd} index. Is there a complete characterization of the {\cd} index in this regard?
    
    \item To provide a more practical direction, our study encourages meticulous experimental exploration to verify whether the {\cd} manipulation problems are difficult to solve in practice.

    \item Returning to our theoretical focus, it is intriguing to study other meaningful parameterizations. Notably, our results for {\mergemanipulation} remain applicable even when the compatibility graph comprises a matching, thereby ruling out fixed-parameter tractability of the problem concerning the treewidth of the compatibility graph. Future investigations could explore the potential influence of the treewidth (or other parameters) of the citation digraph on the parameterized complexity of this problem.

    \item Solving the open questions listed in this paper also presents a promising avenue for further research.
\end{itemize}

\section*{Acknowledgments}
The author thanks the anonymous reviewers of IJCAI~2024 for their valuable feedback and insightful comments on a preliminary version of this work.

%
\end{document}